\pdfoutput=1
\documentclass[12pt,hidelinks]{article}

\usepackage{amsmath}
\usepackage{amsthm}
\usepackage{thmtools}
\usepackage{tikz}
\usepackage[shortlabels]{enumitem}
\usepackage{enumitem}
\usepackage{booktabs}
\usepackage[title]{appendix}

\usepackage{algorithm}
\usepackage[noend]{algpseudocode}

\usepackage[numbers,sort,compress]{natbib}

\usepackage{color}
\definecolor{darkblue}{rgb}{0.0,0.0,0.2}
\usepackage[colorlinks=true,breaklinks=true,bookmarks=true,urlcolor=darkblue,
  citecolor=darkblue,linkcolor=darkblue,bookmarksopen=false,draft=false]{hyperref}
\usepackage[colorinlistoftodos,textsize=scriptsize]{todonotes}
  
\newcommand{\Comments}{0}
\newcommand{\mynote}[2]{\ifnum\Comments=1\textcolor{#1}{#2}\fi}
\newcommand{\mytodo}[2]{\ifnum\Comments=1
	\todo[linecolor=#1!80!black,backgroundcolor=#1,bordercolor=#1!80!black]{#2}\fi}

\newcommand{\btw}[1]{\mytodo{gray!10!white}{\textcolor{gray!40!black}{BTW: #1}}}
\ifnum\Comments=1               
\usepackage[top=1in,bottom=1in,left=1in,right=1.4in]{geometry}
\else
\usepackage[top=1in,bottom=1in,left=1.2in,right=1.2in]{geometry}
\fi

\usetikzlibrary{positioning, arrows.meta, shapes, calc, fit, math, bending}

\setenumerate{itemsep=-.5ex}
\setitemize{itemsep=-.5ex}

\def\isLeft{1}
\def\isCw{1}
\def\centerOffset{1}

\tikzset{
    arc center/.is choice,
    arc center/left/.code={\def\isLeft{1}\tikzset{arc to}},
    arc center/right/.code={\def\isLeft{0}\tikzset{arc to}},
    arc center offset/.code={\def\centerOffset{#1}\tikzset{arc to}},
    arc direction/.is choice,
    arc direction/cw/.code={\def\isCw{1}\tikzset{arc to}},
    arc direction/ccw/.code={\def\isCw{0}\tikzset{arc to}},
    arc to/.style={
        to path={
            [evaluate={
                coordinate \u, \v, \m, \mv, \mvOrtho, \c, \cu, \cv;
                \u = (\tikztostart);
                \v = (\tikztotarget);
                \m = (\u)!0.5!(\v);
                \mv = (\v)-(\m);
                if \isLeft then {
                    \mvOrtho = (-\mvy, \mvx);
                } else {
                    \mvOrtho = (\mvy, -\mvx);
                };
                \c = (\m)+(\mvOrtho);
                \c = (\m)!\centerOffset!(\c);
                \cu = (\u)-(\c);
                \cv = (\v)-(\c);
                \radius = veclen(\cu);
                \startAngle = atan2(\cuy, \cux);
                \endAngle = atan2(\cvy, \cvx);
                if \isCw then {
                    if \startAngle < \endAngle then {
                        \endAngle = \endAngle - 360;
                    };
                } else {
                    if \startAngle > \endAngle then {
                        \endAngle = \endAngle + 360;
                    };
                };
            }]
            arc[radius=\radius pt, start angle=\startAngle, end angle=\endAngle] \tikztonodes
        }
    },
    vertex/.style={circle, inner sep=1pt},
    left loop/.style={arc center=left, arc direction=cw, arc center offset=10pt},
    right loop/.style={arc center=right, arc direction=ccw, arc center offset=10pt},
    invisible/.style={text opacity=0},
    -{Stealth[round, flex=1]}
}

\usepackage{xspace}

\newcommand{\Lc}{\mathcal{L}}  
\newcommand{\Bc}{\mathcal{B}}  
\newcommand{\Ac}{\mathcal{A}}
\newcommand{\Fc}{\mathcal{F}}

\newcommand{\shift}[1]{\mathsf{X}_{#1}} 
\newcommand{\term}[1]{\textit{#1}}
\newcommand{\leftm}{{\vartriangleright}}  
\newcommand{\rightm}{{\vartriangleleft}} 
\newcommand{\Cerny}{\v{C}ern\'{y}\xspace}
\newcommand{\PSPACE}{\ensuremath{\mathsf{PSPACE}}\xspace}

\newcommand{\cc}[1]{\ensuremath{\mathsf{#1}}\xspace}

\newcommand{\Zb}{\mathbb{Z}}

\newcommand{\problem}[1]{\textsc{#1}}
\newcommand{\irredp}{\problem{Irreducibility}\xspace}
\newcommand{\equalp}{\problem{Equality}\xspace}
\newcommand{\inclusionp}{\problem{Subshift}\xspace}
\newcommand{\sftp}{\problem{Sft}\xspace}
\newcommand{\sdpp}{\problem{$\exists$Sdp}\xspace}
\newcommand{\minimalp}{\problem{Minimality}\xspace}
\newcommand{\syncp}{\problem{SyncWord}\xspace}
\newcommand{\dfacapp}{\problem{DfaInt}\xspace}
\newcommand{\dfacupp}{\problem{DfaUnion}\xspace}

\def\ygap{2}
\def\xgap{2}
\def\arrowgap{.125}

\usepackage[bitstream-charter, cal=cmcal]{mathdesign}

\declaretheorem[style=definition, numberwithin=section]{theorem}
\declaretheorem[style=definition, sibling=theorem]{lemma}

\declaretheorem[style=definition, sibling=theorem]{proposition}

\declaretheorem[style=definition, sibling=theorem]{corollary}
\declaretheorem[style=remark, sibling=theorem]{remark}

\declaretheoremstyle[
qed=\qedsymbol,
]{reductionstyle}
\declaretheorem[style=reductionstyle]{reduction}

\renewcommand{\in}{\smallin}
\renewcommand{\notin}{\notsmallin}

\title{Computational complexity of problems for deterministic presentations of sofic shifts}

\author{Justin Cai
  \\{\small University of Colorado Boulder}
  \\\texttt{\small justin.cai@colorado.edu}
  \and
  Rafael Frongillo
  \\{\small University of Colorado Boulder}
  \\\texttt{\small raf@colorado.edu}}

\begin{document}

\maketitle

\begin{abstract}
  Sofic shifts are symbolic dynamical systems defined by the set of bi-infinite sequences on an edge-labeled directed graph, called a presentation.
  We study the computational complexity of an array of natural decision problems about presentations of sofic shifts, such as whether a given graph presents a shift of finite type, or an irreducible shift; whether one graph presents a subshift of another; and whether a given presentation is minimal, or has a synchronizing word.
  Leveraging connections to automata theory, we first observe that these problems are all decidable in polynomial time when the given presentation is irreducible (strongly connected), via algorithms both known and novel to this work.
  For the general (reducible) case, however, we show they are all \cc{PSPACE}-complete.
  All but one of these problems (subshift) remain polynomial-time solvable when restricting to synchronizing deterministic presentations.
  We also study the size of synchronizing words and synchronizing deterministic presentations.
  \btw{Old abstracts commented out}
\end{abstract}


\section{Introduction}

Symbolic dynamics in dimension one is the study of shift spaces, which are topological dynamical systems given by ``shifting'' bi-infinite sequences of symbols.
Sofic shifts are shift spaces whose points are given by the label sequences for bi-infinite walks in a labeled graph, called a presentation.
As they characterize the factors of subshifts of finite type (SFTs), sofic shifts have fundamental importance in symbolic dynamics.
They also have an array of applications both within and outside of dynamical systems, including billiards, ergodic theory, continuous dynamics, and information theory, automata theory, and matrix theory~\cite{marcus2008symbolic}.
In particular, one motivation for the present work is the set of computational problems that arise in application to continuous maps via Conley index theory~\cite{kwapisz2000cocyclic,kwapisz2004transfer,day2019sofic}.

Despite their fundamental importance, however, many basic questions about the computational complexity of sofic shifts remain open.
In Table~\ref{tab:problems}, we give seven natural decision problems, many of which also arise frequently in applications.
For example, given a labeled graph, does it present an SFT?
Does a given reducible presentation actually present an irreducible sofic shift?
Do two given labeled graphs present the same shift?
For special cases, such when the given presentations are irreducible, some of these problems are known to be in \cc{P}, i.e., they admit a polynomial-time algorithm.
For the general case, however, only the complexity of \syncp is known: it is \cc{PSPACE}-complete to determine whether a given deterministic presentation has a synchronizing word~\cite{berlinkov2014two}.

In this work, we resolve the complexity of the remaining six problems, showing that they are all \cc{PSPACE}-complete in the general case.
We also study two special cases, sofic shifts given by deterministic presentations which are either irreducible or synchronizing.
For these cases, the problems are generally in \cc{P}.
In fact, the only exception is \inclusionp for synchronizing deterministic presentations, which is again \cc{PSPACE}-complete.
Our reductions also shed light on the size of the smallest synchronizing word and minimal synchronizing deterministic presentation, namely that both can be exponentially large in the given presentation.
These results are significant for understanding sofic shifts in their own right, as well as relevant for applications.

\begin{table}[t]
  \centering
  \begin{tabular}{lll}
    \toprule
    Problem & Input & Decision \\
    \midrule
    \irredp & $G$ & Is $\shift{G}$ irreducible? \\
    \equalp & $G,H$ & Is $\shift{G} = \shift{H}$? \\
    \inclusionp & $G,H$ & Is $\shift{G} \subseteq \shift{H}$? \\
    \sftp & $G$ & Is $\shift{G}$ an SFT? \\
    \sdpp & $G$ & Does $\shift{G}$ have an SDP? \\
    \minimalp & $G,k$ & Does $\shift{G}$ have a $k$-vertex deterministic presentation? \\
    \syncp & $G$ & Does $G$ have a synchronizing word? \\
    \bottomrule
  \end{tabular}
  \caption{Natural decision problems for sofic shifts.  We show that all are \PSPACE-hard in the general case; see Table~\ref{tab:results-overview} for an overview of our results.  For the inputs to the problems, $G,H$ are deterministic presentations, and $k$ is a positive integer.}
  \label{tab:problems}
\end{table}

\subsection{Relation to the literature}
\label{sec:lit}

\paragraph{Conjugacy}

Absent from our list of problems is arguably the most important: deciding whether two sofic shifts are isomorphic, or \term{conjugate}.
In general, the decidability of this conjugacy problem is open, even when restricting to the class of SFTs.
Verifying the natural certificates of
conjugacy, known as sliding block codes, is computable in polynomial time for SFTs given by vertex shifts, and deciding if there is a
certificate of a fixed size is \cc{GI}-hard, meaning there is a polynomial-time reduction from the graph isomorphism problem to that problem~\cite{schrock2021computational}.
One can partially decide
nonconjugacy via conjugacy invariants, i.e., properties
which isomorphic objects share.
For dynamical systems, a set of invariants related to the
connectivity of the state space are being
(topologically)
transitive, mixing, and nonwandering.
For shift spaces, topological transitivity is
also known as irreducibility. 
The class of SFTs is contained with the class of sofic shifts,
and the class of SFTs is closed under conjugacy; thus, for
sofic shifts, being an SFT is a conjugacy invariant.
\citet{conen1977finite}
show that all the above invariants are decidable.
We show that for sofic
shifts given by deterministic presentations, deciding
if they are irreducible or an SFT is \cc{PSPACE}-complete.
The \cc{PSPACE}-hardness of being mixing or nondwandering
also follows immediately from our reduction.
Interestingly, it also follows from our reduction that
deciding conjugacy of sofic shifts is at least \cc{PSPACE}-hard.

\paragraph{Automata theory}

Sofic shifts have a close relationship with automata theory.
A finite automaton (FA) roughly corresponds to an edge-labeled
directed graph with a set of initial and accepting states.
\footnote{We define an FA the same way some authors define a \emph{nondeterminisc finite automaton (NFA)}; we use ``FA'' to avoid confusion when we consider deterministic finite automata (DFAs) as a subset of FAs.}
The language
of an FA is the set of words labeled by a path starting at an initial state
and ending at an accepting state. The languages described by finite
automata are known as \term{regular} languages.
Similarly, we define the language of a shift space to be the set of finite words appearing in
points of the shift space. 

A basic result about shift spaces says that shift spaces are
determined by their languages: two shift spaces are equal if
and only if their languages are equal.
Furthermore, the
languages of sofic shifts are regular in the above sense.\footnote{The languages of sofic shifts are exactly the languages which are \term{factorial}, \term{prolongable}, and regular. \cite[Proposition 1.3.4]{lind2021introduction}}
More specifically,
by interpreting a presentation of a sofic shift as an FA where every state
is both initial and final, then the language of a
presentation (interpreted as an FA) is the same as the
language of the sofic shift it presents. These connections allows
us to use automata-theoretic tools to study sofic shifts.

A \term{deterministic finite automaton} (DFA)
over an alphabet $\Sigma$ is an FA with a single initial state,
such that at each state $q$ and
for each $a \in \Sigma$, there is exactly one edge leaving $q$ labeled $a$.
A deterministic presentation of a sofic shift, when thought of as an FA, has a similar definition: a presentation is \term{deterministic}
if for each $a \in \Sigma$, there is \emph{at most} one edge leaving
that state labeled $a$. Comparing the languages of two DFAs
(i.e.\ whether their languages are equal, or if one is a subset of the other) is computable
in polynomial time.
However, for FAs in general, the same problem is \cc{PSPACE}-complete
\cite{malcher2004minimizing}. 
Comparing the languages of presentations in general is also \cc{PSPACE}-complete,
as a corollary of results from \citet{czeizler2006testing}.
The question that remains is therefore the complexity of comparing languages of deterministic presentations.
It is likely known that comparing languages of irreducible (i.e.\ strongly connected)
deterministic presentations is computable in polynomial time;
in Section~\ref{sec:testing-subshift}, we give an algorithm.
For
deterministic presentations in general, we show in Section~\ref{sec:hardn-equal-cont} that comparing
languages (and thus comparing the shift spaces they present)
of this type of FA is \cc{PSPACE}-complete.

\paragraph{Minimization of presentations}

We say two FAs are \term{equivalent} if they have the
same language.
Algorithms for \term{minimizing} a DFA, i.e.\ finding an equivalent DFA with
fewer states, have been well-studied.
This minimization problem has nice properties: every DFA has an unique minimal
equivalent DFA which can be computed in polynomial time.  For FAs in
general, minimization is not as nice: FAs do not necessarily have
unique minimal equivalent FA \cite{arnold1992note}, and deciding if there is an FA with fewer
states than a given FA is \cc{PSPACE}-complete \cite{malcher2004minimizing}.

As a class of FAs,
irreducible deterministic presentations of sofic shifts have similar minimality
properties to DFAs: every irreducible deterministic presentation has a
unique minimal equivalent irreducible deterministic presentation which
is computable in polynomial time~\cite{lind2021introduction}.
(Furthermore, the property
characterizing a minimal irreducible deterministic presentation
in a sense is exactly the same property as a minimal DFA.)
This observation leaves the question: do
reducible deterministic presentations share the minimzation properties
of the class of DFAs or of the class of FAs?
The class of FAs arguably closest to general deterministic presentations are the \term{multiple-entry DFAs} (mDFAs), which are
essentially DFAs with one source of nondeterminism: multiple initial states.
A deterministic
presentation can be made into an equivalent mDFA by by adding a sink
state, thus ``fully determinizing'' every state, and then interpreting
every state but the sink state as an initial and final state
(c.f.\ \term{sink vertex graph} in Section \ref{sec:find-synchr-words}).
As mDFAs share the minimization
properties of the class of FAs~\cite{malcher2004minimizing, holzer2001state}, and
general deterministic presentations do not have unique minimal equivalent deterministic presentations~\cite{jonoska1996sofic}, one may suspect that minimization of general deterministic presentations is \cc{PSPACE}-complete.
Indeed, we show this result in Section~\ref{sec:sft-reduct}.

\paragraph{Synchronizing words}

Another equivalent way of defining a DFA is by specifying a
transition function; i.e., a function $\delta\colon Q \times \Sigma \to Q $ for some set of states $Q$ and finite alphabet $\Sigma$.
The transition function then
naturally extends to a function $\delta\colon Q \times \Sigma^* \to Q$ from the states and words over an alphabet.
A \term{synchronizing word} (also called a \term{reset word}) for a DFA is
a word that transitions every state to a single state:
$w$ is synchronizing if $\delta(p,w) = \delta(q,w)$ for all states
$p$ and $q$.
(Equivalently, the function $q \mapsto \delta(q, w)$ 
is a constant function.)
A deterministic
presentation can be seen as a DFA with a partially defined transition
function (a function whose domain is a subset of $Q \times \Sigma$); call a DFA with a partial transition function a \term{partial} DFA.
There are multiple ways to generalize the notion of a synchronizing word
to partial DFAs, for example, a \term{carefully} synchronizing word \cite{shabana2020careful} and
a \term{exact} synchronizing word \cite{shabana2019exact, travers2011exact}.
The former is a word whose
transition is defined at all states and sends all states to a single state;
the latter is a word whose transition is defined at least one state 
and sends every state (where it is defined) to the same state.
Interpreting a deterministic presentation as a partial DFA, a
synchronizing word for a deterministic presentation of a sofic shift is defined as an exact synchronizing word.
Note that a DFA might be called
synchronizing or synchronized if it has a synchronizing word; for
presentations of sofic shifts, our usage of a synchronizing presentation
corresponds to that of \citet{jonoska1996sofic}: for every state $q\in Q$, there is a synchronizing word that sends every state to $q$.

For DFAs, one can decide if a synchronizing word exists (and find one) in polynomial time
via Eppstein's algorithm~\cite{eppstein1990reset}.
Independently, \citet{travers2011exact} gave a similar algorithm
which can be used to determine if a synchronizing word exists in
an irreducible deterministic presentation.
In Section \ref{sec:find-synchr-words}, we describe an
algorithm to find a synchronizing word in an irreducible
deterministic presentations (Algorithm \ref{alg:sw}) which combines the techniques
of the previously two mentioned algorithms.
We extend these techniques to subshift testing
for irreducible deterministic presentations (Algorithm \ref{alg:subshift})
and then use  synchronizing word algorithm and subshift testing algorithm
together as subprocedures for testing if a deterministic presentation is synchronizing.
For deterministic presentations in general,
deciding if a synchronizing word exists is \cc{PSPACE}-complete,
and the size of a minimum length synchronizing word may be exponentially
large with respect to the number of states.
While both of these facts were already implied by
\citet{berlinkov2014two}, for completeness we provide proofs in Sections~\ref{sec:hard-sync-word} and~\ref{sec:short-synchr-word}, respectively.

\paragraph{Synchronizing deterministic presentations}

\citet{jonoska1996sofic} introduced synchronizing deterministic
presentations, as defined above: for every state $q\in Q$, there is a synchronizing word that sends every state to $q$.
The shift spaces given by synchronizing deterministic presentations slightly
generalize those given by irreducible deterministic presentations while
retaining serveral nice properties.
For example, synchronizing deterministic presentations
share the minimization properties of the class of DFAs:
every synchronizing deterministic presentation has a unique
minimal equivalent synchronizing deterministic presentation
that is computable in polynomial time. For irreducible
deterministic presentations, it is known that \equalp and \sftp
are in \cc{P}.
We show that the algorithms for the irreducible
case generalize cleanly to the synchronizing case, implying that
\equalp and \sftp are in \cc{P} for synchronizing determinstic presentations
(Sections~\ref{sec:iso-eq} and~\ref{sec:sft-test-synchr}).
Interestingly, although \inclusionp is in \cc{P} for irreducible
deterministic presentations, \inclusionp is \cc{PSPACE}-complete
for synchronizing deterministic presentations (Remark~\ref{remark:subshift-sdp-hard}).

Not all sofic shifts have synchronizing deterministic
presentations.
In fact, we show that the problem of deciding whether a
sofic shift has a synchronizing deterministic presentation, \sdpp, is
\cc{PSPACE}-complete (Section~\ref{sec:hardn-equal-cont}).
For irreducible sofic shifts, minimal
synchronizing deterministic presentations and minimal deterministic
presentations are the same.
For reducible sofic shifts, however, these
minimal presentations are not necessarily the same.
Indeed, we show that a
minimal synchronizing deterministic one can be exponentially larger
than a minimal deterministic one (Section~\ref{sec:minim-synchr-determ}).

\section{Background and Setting}

\subsection{Shift spaces and presentations}

Here, we introduce basic notions about shift spaces, sofic shifts, and presentations. Definitions and notation follow \citet{lind2021introduction}.

Let $\Sigma$ be a finite set. We refer to a finite sequence as a \term{word}, and we denote by $\Sigma^*$ the set of words over $\Sigma$. A subset of $\Sigma^*$ is called a \term{language}.  For $w \in \Sigma^*$, we denote the length of $w$ as $|w|$.  We denote the empty word as $\epsilon$, and note that $|\epsilon| = 0$ and $\epsilon \in \Sigma^*$. The \term{full $\Sigma$-shift} is the set $\Sigma^\Zb$ of bi-infinite sequences over $\Sigma$.  Let $x \in \Sigma^\Zb$. For $i \leq j$, we denote $x_{[i,j]} \triangleq x_i x_{i+1} \dots x_j$. We say a word $w$ \term{appears} in $x$ if there are $i$ and $j$ with $x_{[i,j]} = w$.  For a collection of words $\Fc \subseteq \Sigma^*$, we define $\shift{\Fc} \triangleq \{\, x \in \Sigma^\Zb : \text{no word in } \Fc \text{ appears in } x \,\}$.  A \term{shift space} is a subset $X \subseteq \Sigma^\Zb$ of the full $\Sigma$-shift such that there is a collection of words $\Fc$ with $X = \shift{\Fc}$.  A \term{shift of finite type (SFT)} is a shift space $X = \shift{\Fc}$ for some finite set $\Fc$.

For a subset $X \subseteq \Sigma^\Zb$ of the full $\Sigma$-shift, we define the \term{language} $\Bc(X)$ of $X$ to be the set of words that appear in some $x \in X$, i.e.,\ $\Bc(X) \triangleq \{\, x_{[i,j]} : x \in X, i \leq j \,\}$.
Shift spaces are characterized by their languages: for every shift
space $X \subseteq \Sigma^\Zb$, one has that $X = \shift{\Sigma^*\setminus\Bc(X)}$. Thus,
for shift spaces $X$ and $Y$, if $\Bc(X) = \Bc(Y)$, then
$X = Y$ \cite[Proposition 1.3.4]{lind2021introduction}.
Additionally, one can easily show inclusion is also respected:
$\Bc(X) \subseteq \Bc(Y)$ if and only if $X \subseteq Y$.
Finally, for $u \in \Sigma^*$, we define the \term{follower set} of $u$
as the set $F_X(u) \triangleq \{\, w \in \Sigma^* : uw \in \Bc(X) \,\}$.

To define sofic shifts, we will work with edge-labeled, directed multigraphs,
where self loops and multiple edges between vertices are permitted.  Formally, a
\term{labeled graph} $G$ consists of a finite set $Q$ of \term{vertices} (or
\term{states}), a finite set $E$ of \term{edges}, functions $i\colon E \to Q$
and $t\colon E \to Q$, assigning each edge an \term{initial} and \term{terminal}
vertex, and a function $\Lc\colon E \to \Sigma$, assigning each edge a
\term{label}.
For a given graph $G$, the symbols $Q_G$, $E_G$, $i_G$, $t_G$,
and $\Lc_G$ will refer to the above sets and functions for the graph $G$.
Additionally, we define the \term{alphabet of $G$} as the set
$\Ac_G$ of labels appearing on edges in $G$ (i.e.\
$\Ac_G \triangleq \Lc_G(E_G)$).
When the labels are irrelevant, we will sometimes call a labeled graph
a graph.

If $G$ is a labeled graph and $P \subseteq Q$ is a subset of vertices,
then the \term{subgraph induced by $P$} (in $G$) is the labeled graph
$H$ given by $Q_H, E_H, i_H, t_H, \Lc_H$, where:
$Q_H \triangleq P$;
$E_H \triangleq \{\, e \in E_G : i_G(e) \in P, t_G(e) \in P\,\}$;
$i_H$,$t_H$ are $i_G,t_G$ restricted to $E_H$;
and $\Lc_H$ is $\Lc_G$ restricted to
$Q_H$.

Let $G$ be a labeled graph.
A \term{path} in $G$ is a finite sequence $\pi=e_1 \dots e_n$ of edges
with $t_G(e_{i}) = i_G(e_{i+1})$ for $i=1,\dots, n-1$.  We assign
$i_G(\pi) \triangleq e_1$ and $t_G(\pi) \triangleq e_n$, and say $\pi$
\term{starts} at $i_G(\pi)$ and \term{ends} at $t_G(\pi)$. Additionally,
we assign $\Lc_G(\pi) \triangleq \Lc_G(e_1) \dots \Lc_G(e_n)$, and say
$\Lc_G(\pi)$ is the \term{label} of $\pi$.  Similarly, a \term{bi-infinite path} in $G$ is a
bi-infinite sequence $x \in E_G^\Zb$ of edges with $t_G(x_i) = i_G(x_{i+1})$
for all $i \in \Zb$. For a bi-infinite path $x$ in $G$, we assign the
\term{label} of $x$ as the bi-infinite sequence
$\Lc_G(x) \in \Ac_G^\Zb$ with
$\Lc_G(x)_i \triangleq \Lc_G(x_i)$.  For a vertex, $q$ we define the
\term{follower set} of $q$ in $G$ as the set
$F_G(q) \triangleq \{\, \Lc_G(\pi) : \pi \text{ is a path in } G \text{
  starting at } q \,\}$.\footnote{
  Under this definition, the empty path $\epsilon$, i.e., the empty sequence of edges, is a valid path in $G$ but one where $i_G$ and $t_G$ are undefined.
  To rectify this omission, for every vertex $q \in Q_G$, we declare $\epsilon_q$ to be a path such that the length of $\epsilon_q$ is $0$, $\epsilon_q$ starts and ends at $q$, and $\Lc(\epsilon_q) = \epsilon$.
}

We now have the necessary definitions to define sofic shifts.
For a labeled graph $G$,
we assign it the shift space
\[\shift{G} \triangleq \{\, \Lc_G(x) : x \text{ is a bi-infinite
  path in } G \,\}.\]  A \term{sofic shift} is a shift space
$X$ such that $X = \shift{G}$ for some labeled
graph $G$,
and we say $G$ is a \term{presentation} of $X$ and that $X$
is the sofic shift \term{presented} by $G$.
For a proof that $\shift{G}$ is actually a shift space, see \citet[Theorem 3.1.4]{lind2021introduction}.

Let $G$ be a labeled graph.
For a given path $\pi$ in $G$, it could be the case that $\Lc_G(\pi)$
is not in the language of $\shift{G}$, as $\pi$ might not appear
in a bi-infinite path.
We say a vertex $q$ in $G$ is \term{stranded}
if there is no edge starting at $q$ or if there is no edge ending at
$q$.  If no vertex is stranded, then we say $G$ is
\term{essential}. When $G$ is essential, every path appears in a
bi-infinite path, so $\Lc_G(\pi)$ is always in the language of
$\shift{G}$.
If one removes a stranded vertex from a presentation, the
sofic shift presented by the resulting presentation is the same as the
one presented by the original presentation.
Thus, every sofic shift has an essential presentation, which can be obtained by iteratively
removing stranded vertices until no more exist \cite[Proposition 2.2.10]{lind2021introduction}.
We therefore make the
following convention: a \term{presentation} refers to an essential labeled
graph.
We will still refer to labeled graphs as being
potentially nonessential; the distinction is needed for the algorithms in
Section \ref{sec:compl-upper-bounds}, where we may call our algorithms on
nonessential graphs (line~\ref{alg:is-sync:line:essential} of
  Algorithm~\ref{alg:sync-test}).

Let $G$ be a labeled graph.  We
say $G$ is \term{deterministic} (also called \term{right-resolving}) if for
every vertex $q$ and every $a \in \Ac_G$, there is at most one edge
labeled $a$ starting at $q$.
If for every vertex $q$ and
$a \in \Ac_G$ there is exactly one edge labeled $a$ starting at $q$,
we say $G$ is \term{fully deterministic}.
If $G$ is deterministic,
one can show by induction that for every vertex $q$ and word $w$, if
$\pi$ is a path starting at a vertex $q$ and $\Lc_G(\pi) = w$, then
$\pi$ is the unique path starting at $q$ with $\Lc_G(\pi) = w$. This observation
motivates the following definition: if there is some path $\pi$
starting at $q$ with $\Lc_G(\pi) = w$, we define
$q \cdot w \triangleq t(\pi)$, otherwise, if there is no such $\pi$,
we leave $q \cdot w$ undefined.
We call $\,\cdot\,$ the \term{transition action}.
Because of determinism, the transition action is
a well-defined partial operation between the vertices of $G$ and
words over the alphabet of $G$, and 
$q \cdot w$ is defined if and only if $w \in
F_G(q)$.
The transition action satisfies the following useful properties, for any state $q\in Q_G$ and words $u,v \in \Ac_G^*$:
\begin{enumerate}[(i)]
\item $uv \in F_G(q)$ if and only if $u \in F_G(q)$ and
  $v \in F_G(q \cdot u)$;
\item if $uv \in F_G(q)$, then $q \cdot uv = (q \cdot u) \cdot v$.
\end{enumerate}
When $G$ is fully deterministic, $q \cdot w$ is defined for all $q$
and $w \in \Ac_G^*$, as then $F_G(q) = \Ac_G^*$.

The transition
action naturally extends to a total operation between subsets of
vertices of $G$ and words over the alphabet of $G$: for each subset
$S \subseteq Q_G$ of vertices and word $w$, we set
$S \cdot w \triangleq \{\, q \cdot w : q \in S, w \in F_G(q) \,\}$.
Every sofic shift therefore has a deterministic presentation \cite[Theorem
3.3.2]{lind2021introduction}, using the same idea as
the subset construction from automata theory~\cite{kozen2012automata}.
By definition, the transition action on subsets is monotonic and
distributes over union:
$S \subseteq T$ implies $S \cdot w \subseteq T \cdot w$
and $(S \cup T) \cdot w = (S \cdot w) \cup (T \cdot w)$
for all $S$, $T$, and $w$.

\subsection{Types of presentations}

Let $G$ be a deterministic labeled graph and let $w$ be a word. We say $w$ is
\term{synchronizing} for $G$ if $Q_G \cdot w = \{ r \}$ for some vertex
$r \in Q_G$. In this case, we say $w$ synchronizes to $r$ (in $G$).
We say a vertex $q$ is synchronizing if there is a word that
synchronizes to $q$.
We say $G$ is synchronizing if every vertex in $G$ is synchronizing.
Let $X$ be a shift space. An
\term{intrinsically synchronizing word} $w$ for $X$ is a word
$w \in \Bc(X)$ such that whenever $uw, wv \in \Bc(X)$, then
$uwv \in \Bc(X)$. If $w$ is synchronizing for $G$, then $w$ is
intrinsically synchronizing for $\shift{G}$, but the
converse need not hold; see Lemma~\ref{lem:sync-word-correspondence}.

Let $X$ be a shift space. We say $X$ is \term{irreducible} if for
every $u, v \in \Bc(X)$, there is a word $w$ such that
$uwv \in \Bc(X)$; if $X$ is not irreducible, then we say $X$ is
\term{reducible}. For a graph $G$, we say $G$
is \term{irreducible} (or \term{strongly connected}) if for every pair of vertices $p$ and $q$, there is a
path starting at $p$ and ending at $q$.
If $G$ is not irreducible, we say $G$ is reducible.
One can easily show that if $G$
is irreducible, then $\shift{G}$ is irreducible. However, $\shift{G}$ may be irreducible even if $G$ is reducible. (See Figure~\ref{fig:example}.)

Let $G$ be a graph, and let $p$ and $q$ be vertices in $G$. We say $q$
is \term{reachable} from $p$ if there is a path starting at $p$ and
ending at $q$.
Under the equivalence relation where $p \approx q$ when $q$ is reachable from $p$ and $q$ is reachable from $q$, the equivalence
classes are called \term{irreducible components} as the subgraphs
induced by them are irreducible.  We say an irreducible
component $C$ is \term{initial} if whenever $q$ is reachable from $p$
and $q \in C$, then $p \in C$. Dually, we say a irreducible component
$C$ is \term{terminal} if whenever if $q$ is reachable from $p$ and
$p \in C$, then $q \in C$.

Let $G$ be a labeled graph. We say two
vertices $p$ and $q$ in $G$ are \term{follower-equivalent} if
$F_G(p) = F_G(q)$, an equivalence relation $\sim$.
We say $G$ is \term{follower-separated} if no
distinct pair of vertices are follower equivalent. 
Given a
labeled graph $G$, the \term{follower-separation} of $G$ is the the
labeled graph $G/{\sim}$ whose vertices are the follower-equivalence
classes of $G$ and with exactly one edge labeled $a$ between two
classes $C_1$ and $C_2$ if and only if there is an edge labeled $a$ in
$G$ from a vertex in $C_1$ to a vertex in $C_2$. Informally, the
follower-separation of $G$ collapses vertices in a given
follower-equivalence class into a single vertex.  The
follower-separation of $G$ enjoys the following properties: we have
$G/{\sim}$ is follower-separated and $\shift{G} = \shift{G/{\sim}}$;
if $G$ is deterministic, then $G/{\sim}$ is deterministic; if $G$ is
essential, then $G/{\sim}$ is essential \cite[Lemma 3.3.8]{lind2021introduction};
if $G$ is synchronizing, then $G/{\sim}$ is synchronizing
\cite[Proposition 4.3]{jonoska1996sofic}.
In particular, every sofic shift has a follower-separated, deterministic
presentation.

The notion of follower-equivalence is similar to the notion of
equivalent states in deterministic finite automata (DFA; see Sections \ref{sec:lit} and \ref{sec:complexity-lower-bounds}).
In fact, one may reduce the problem of computing
follower-equivalence to computing equivalent states in DFA, as follows.
Add a ``sink'' state to $G$, and edges to
the sink state to make $G$ fully deterministic (c.f. Section \ref{sec:find-synchr-words}).
Now consider the resulting graph as a DFA, with an arbitrary initial state, and where every state but the sink state is an accepting state.
One can show that two
states in $G$ are follower-equivalent if and only if they are
equivalent as states in the constructed DFA.
Therefore, one can use
Hopcroft's algorithm for state equivalence in DFAs to
compute follower-equivalences in polynomial time~\cite{hopcroft1971n}.

\subsection{Basic results}

In this section, we discuss several useful facts which we use throughout the paper.
To begin, the following statements about a deterministic presentation $G$ establish
basic relationships between its transition action, follower sets $F_G$ and $F_{\shift{G}}$, the language $\Bc(\shift{G})$, and synchronizing words for $G$.
\footnote{Recall that a presentation refers to
  an \textit{essential} labeled graph, a necessary condition
  here, as these statements do not necessarily
  hold if $G$ is nonessential.}
The statements follow immediatly from the definitions.

\begin{proposition}\label{prop:basic}
  Let $G$ be a deterministic presentation.
  Then, we have
  \begin{enumerate}[(i)]
  \item $w \in \Bc(\shift{G})$ if and only if $Q_G \cdot w \neq \varnothing$;
  \item $\Bc(\shift{G}) = \bigcup_{q \in Q} F_G(q)$;
  \item $F_{\shift{G}}(w) = \bigcup_{q \in Q \cdot w} F_G(q)$;
  \item if $w$ is synchronizes to $r$ in $G$, then $F_{\shift{G}}(w) = F_G(r)$;
  \item if $w$ is intrinsically synchronizing for $\shift{G}$ and $w \in F_{\shift{G}}(u)$, then $F_{\shift{G}}(uw) = F_{\shift{G}}(w)$.
  \end{enumerate}
\end{proposition}

Next, we review results of \citet{jonoska1996sofic} about synchronizing deterministic presentations.
First, we state a useful result about
the correspondence between synchronizing and intrinsically synchronizing
words in synchronizing deterministic presentations. For
deterministic presentations in general, only the forward implication
of this result holds. The following is essentially
Proposition 9.5 of \citet{jonoska1996sofic}.

\begin{lemma}\label{lem:sync-word-correspondence}
  Let $G$ be a follower-separated, synchronizing deterministic presentation.
  Then, $w$ is synchronizing
  for $G$ if and only if $w$ is intrinsically synchronizing for $\shift{G}$.
  \footnote{We remind the reader that our notion of synchronization is different from ``careful'' sychronization from automata theory; see ``{Synchronizing words}'' in \S~\ref{sec:lit}.}
\end{lemma}

\begin{proof}
  Suppose $w$ is synchronizing for $G$, and let $uw, wv \in \Bc(\shift{G})$.
  As $w$ is synchronizing for $G$, by Proposition~\ref{prop:basic}(iv),
  it follows that there is a vertex $r$ such
  that $F_{\shift{G}}(uw) = F_G(r)$ and $v \in F_G(r)$. Thus, we have $v \in F_G(uw)$ so
  $uwv \in \Bc(\shift{G})$.

  Conversely, suppose $w$ is intrinsically synchronizing for
  $\shift{G}$.  Let $p$ and $q$ be vertices in $G$ with $w \in F_G(p)$
  and $w \in F_G(q)$. As $G$ is synchronizing, let $u_p$ and $u_q$ be
  words synchronizing to $p$ and $q$ in $G$, respectively.
  We next show that $F_G(p \cdot w) \subseteq F_G(q \cdot w)$.
  Let $v \in F_G(p \cdot w)$, so that $wv \in \Bc(\shift{G})$. As $u_q$
  synchronizes to $q$ in $G$ and $w \in F_G(q)$, we have
  $u_q w \in \Bc(\shift{G})$. As $w$ is intrinsically synchronizing
  for $\shift{G}$, we have $u_q w v \in \Bc(\shift{G})$,
  i.e.  $Q_G \cdot u_q w v \neq \varnothing$.  But as
  $Q_G \cdot u_q w = \{q \cdot w\}$, we have
  $v \in F_G(p \cdot w)$.
  Thus, $F_G(p \cdot w) \subseteq F_G(q \cdot w)$; moreover, the same argument swapping the roles of $p$ and $q$ gives $F_G(q \cdot w) \subseteq F_G(p \cdot w)$ and therefore $F_G(p \cdot w) = F_G(q \cdot w)$.
  As $G$ is follower-separated, we conclude $p \cdot w = q \cdot w$, implying that $w$ is synchronizing for $G$.
\end{proof}

The following characterization of when
a sofic shift has a synchronizing deterministic presentation is slightly modified from Theorem 8.13 and Corollary 9.6 in \citet{jonoska1996sofic}.

\begin{theorem}\label{thm:sdp-char}
  Let $X \subseteq \Sigma^\mathbb{Z}$ be a sofic shift. Then, $X$ has a
  synchronizing deterministic presentation if and only if
  for every $u \in \Bc(X)$ there is an intrinsically
  synchronizing word $w$ for $X$ such that $u \in F_X(w)$.
\end{theorem}

\begin{proof}  
  Let $G$ be a synchronizing deterministic presentation for $X$,
  and let $w \in B(X)$. By Proposition~\ref{prop:basic}(ii), there is a vertex $q$
  such that $w \in F_G(q)$. As $G$ is synchronizing, there is
  a word $u$ that synchronizes to $q$. By Proposition~\ref{prop:basic}(iv), we have
  $F_X(u) = F_G(q)$, so $w \in F_X(u)$.
  By Lemma~\ref{lem:sync-word-correspondence}, we also have that $u$
  is intrinsically synchronizing for $X$.

  Conversely, suppose for every $u \in \Bc(X)$ there is an intrinsically
  synchronizing word $w$ for $X$ such that $u \in F_X(w)$.  Let
  $\mathcal{C}$ be the collection of the follower sets of intrinsically
  synchronizing words for $X$, i.e.
  \[\mathcal{C} \triangleq
    \{\,F_X(w) : w \text{ is intrinsically synchronizing for $X$} \,\}.\]
  This collection is finite since the collection of
  all follower sets of a sofic shift is finite \cite[Theorem 3.2.10]{lind2021introduction}.
  We will construct a
  synchronizing deterministic presentation $G$ whose vertex set is
  $\mathcal{C}$. For each $a \in \Sigma$ and $F_X(w) \in \mathcal{C}$,
  if $a \in F_X(w)$, add an edge labeled $a$ from $F_X(w)$ to
  $F_X(wa)$. This definition is well-defined, i.e., does not depend on the choice of $w$, by the following two facts, both
  assuming $a \in F_X(w)$: if $F_X(w) = F_X(w')$, then $F_X(wa) = F_X(w'a)$, and
  $wa$ is intrinsically synchronizing (so that
  $F_X(wa) \in \mathcal{C}$).  By construction, $G$ is deterministic.
  One can also establish the following properties of $G$:
  for $F_X(u) \in \mathcal{C}$, we have $F_G(F_X(u)) = F_X(u)$, and if
  $w \in F_X(u)$, then $F_X(u) \cdot w = F_X(uw)$.

  We next show that $G$ is synchronizing.
  Let
  $F_X(w) \in \mathcal{C}$, so that $w$ is intrinsicaly synchronizing
  for $X$. We will show that $w$ synchronizes to $F_X(w)$ in $G$. Let
  $F_X(u) \in \mathcal{C}$, and suppose $w \in F_G(F_X(u))$.
  As $w$ is intrinsically synchronizing and $w \in F_X(u)$,
  by Proposition \ref{prop:basic}, we have that $F_X(uw) = F_X(w)$.
  This implies that $F_X(u) \cdot w = F_X(uw) = F_X(w)$. Thus, for any
  $F_X(u) \in \mathcal{C}$ with $w \in F_G(F_X(u))$, we have $F_X(u) \cdot w = F_X(w)$, so
  $w$ synchronizes to $F_X(w)$ in $G$.

  It remains to show $X = \shift{G}$. By
  construction, the follower set of a vertex in $G$ is a follower set
  of a word in $X$, so $\Bc(\shift{G}) \subseteq \Bc(X)$. Conversely, let
  $u \in \Bc(X)$.  By our initial assumption, there is an
  intrinsically synchronizing word $w$ for $X$ with $u \in F_X(w)$. As
  $F_G(F_X(w)) = F_X(w)$, we have $u \in F_G(F_X(w))$, i.e., there is a
  vertex $q$ in $G$ such that $u \in F_G(q) \subseteq \Bc(\shift{G})$. Thus, we have $\Bc(X) \subseteq \Bc(\shift{G})$
  and consequently $\Bc(X) = \Bc(\shift{G})$.
\end{proof}

The next result says when the sofic shift presented by a (possibly
reducible) deterministic presentation is irreducible. The forward implication of
this result follows from Lemma 6.4 of \citet{jonoska1996sofic},
and the
reverse implication follows immediately from the irreducibility of
$H$.

\begin{theorem}\label{thm:irred-char}
  Let $G$ be follower-separated, deterministic presentation. Let
  $H$ be the subgraph induced by the synchronizing vertices of $G$. Then,
  $\shift{G}$ is irreducible if and only if $\shift{G} = \shift{H}$
  and $H$ is induced by a terminal irreducible component.
\end{theorem}

Finally, we state some facts about SFTs. The first is a
characterization of when a shift space is an SFT, and the second is a
sufficient condition for when $\shift{G}$ is an SFT for a presentation
$G$.  Respectively, these correspond to Theorem 2.1.8 and Proposition
2.2.6 of \citet{lind2021introduction}.

\begin{theorem}\label{thm:sft-char}
  A shift space $X$ is an SFT if and only if there exists an
  integer $M \geq 0$ such that every word $w \in \Bc(X)$ with
  $|w| \geq M$ is intrinsically synchronizing for $X$.
\end{theorem}

\begin{lemma}\label{lem:sft-edge}
  Let $G$ be a presentation. If every edge in $G$ is labeled uniquely,
  then $\shift{G}$ is an SFT.
\end{lemma}

\section{Complexity Upper Bounds and Algorithms}
\label{sec:compl-upper-bounds}

In this section we detail polynomial-time algorithms for some problems in Table~\ref{tab:problems}.
In particular, we give polynomial-time algorithms for \sftp and \equalp for synchronizing deterministic presentations, and additionally for \syncp and \inclusionp for irreducible presentations.
We also give a polynomial-time algorithm to test whether a given deterministic presentation is synchronizing.
Some algorithms follow from known results, whereas others, to our knowledge, are novel to this work.

\subsection{Finding synchronizing words}
\label{sec:find-synchr-words}

\citet{eppstein1990reset} gives a polynomial-time algorithm for finding
synchronizing words in fully deterministic presentations.
Here, we show the algorithm can be extended to irreducible deterministic presentations, implying that \syncp is in \cc{P} for such presentations.
As we show in Theorem~\ref{thm:exist-sync-hard}, \syncp is \cc{PSPACE}-complete for general presentations.

\begin{theorem}\label{thm:sw-correct}
  Given an irreducible deterministic labeled graph $G$,
  Algorithm~\ref{alg:sw} returns a synchronizing word for $G$
  if one exists, and nil otherwise.  
\end{theorem}

To prove this result, we first introduce the notion of
pair-synchronizing words. Let $p$ and $q$ be vertices in a
deterministic graph $G$.  We say a word $w$ is
\term{pair-synchronizing} for $p$ and $q$ if $|\{p,q\}\cdot w| = 1$, i.e., if there exists a vertex $r\in Q_G$ such that $\{p,q\}\cdot w = \{r\}$.
This condition breaks into the following three cases:
\begin{enumerate}[(i)]
\item $w \in F_G(p)$ and $w \notin F_G(q)$;
\item $w \notin F_G(p)$ and $w \in F_G(q)$;
\item $w \in F_G(p) \cap F_G(q)$ and $p \cdot w = q \cdot w$.
\end{enumerate}
If $X \subseteq Q_G$ is a subset of vertices with $|X|\geq 2$ and $w$ is pair-synchronizing for distinct $p,q\in X$,
then we have $|X| > |X \cdot w| \geq 1$. This
property motivates Algorithm~\ref{alg:sw}. The algorithm operates by
iteratively building a word $u$ and tracking a subset $X$ of vertices,
maintaining the invariants that $Q_G \cdot u = X$ and $|X| \geq 1$. On
each iteration of the main loop, the algorithm searches for a
pair-synchronizing word $w$ for some pair of distinct vertices in $X$,
and if one is found, then updates $u$ to $uw$ and $X$ to
$X \cdot w$.
The property above ensures that the invariants of
$X$ and $u$ are maintained.
Since $|X|$ must decrease by at least 1 in each iteration, the algorithm returns after at most $|Q_G|$ iterations.

\begin{proof}[Proof of Theorem~\ref{thm:sw-correct}]
  If Algorithm~\ref{alg:sw} returns a non-nil value, it must have exited at
  line~\ref{alg:sw:line:ret-nonnil}, which implies $|X| \leq 1$.
  As the invariant that $|X| \geq 1$ was maintained
  throughout the algorithm, we must have $|Q_G \cdot u| = 1$, so
  the word $u$ that was returned is a synchronizing word for $G$.

  Conversely, if Algorithm~\ref{alg:sw} returned nil, it must have exited at
  line~\ref{alg:sw:line:ret-nil}, which implies that
  there are two distinct vertices such that there is no
  pair-synchronizing word for them.
  Yet, as we show next, if $G$
  has a synchronizing word, then every pair of distinct
  vertices has a pair-synchronizing word.
  Thus, $G$ must not have a synchronizing word.

  Let $p$ and $q$ be distinct vertices in $G$, and suppose $w$
  is a synchronizing word for $G$. As $w$ is synchronizing, there is
  some vertex $s$ with $w \in F_G(s)$. As $G$ is irreducible, there is
  a word $u$ such that $p \cdot u = s$, which implies that
  $uw \in F_G(p)$.  If $uw \notin F_G(q)$, then $uw$ is a
  pair-synchronizing for $p$ and $q$ under case (i) above.
  Otherwise, we have $uw \in F_G(q)$.
  As $w$ is synchronizing, then $p \cdot uw = q \cdot uw$, so
  $uw$ is still pair-synchronizing for $p$ and $q$, under case (iii).
\end{proof}

\begin{algorithm}
  \caption{Finding synchronizing words}
  \label{alg:sw}
  \begin{algorithmic}[1]
    \Require $G$ is a deterministic graph
    \Procedure{synchronizing-word}{$G$}
    \State $X \gets Q_G$;\quad $u \gets \epsilon$
    \While {$|X| \geq 2$}
    \State choose distinct $p, q \in X$
    \State find a word $w$ that is pair-synchronizing for $p$ and $q$
    \If{$w$ exists}
    \State $X \gets X \cdot w$;\quad $u \gets uw$
    \Else
    \State \Return nil \label{alg:sw:line:ret-nil}
    \EndIf
    \EndWhile
    \State \Return u \label{alg:sw:line:ret-nonnil}
    \EndProcedure
  \end{algorithmic}
\end{algorithm}

To implement this Algorithm~\ref{alg:sw} in polynomial time (with
respect to the size of its input $G$), we need a method to compute a
pair-synchronizing word for a given pair of vertices. We give such a method using two auxillary graphs, the first of which encodes what words
are not within a follower set of a vertex, and the second of which
encodes pairs of paths sharing the same label.

If $G$ is a labeled graph and $\Gamma$ is an alphabet, the
\term{sink vertex graph of $G$ with alphabet $\Gamma$} is the graph
$G^0$ constructed as follows.
Start with the graph $G$, and add a new vertex
$0$ to $G^0$.
For every vertex $q$ in $G^0$ and
$\ell \in \Gamma$, add an edge labeled $\ell$ from $q$ to $0$ if
there is no edge labeled $\ell$ starting at $q$.
One can show that for
$w \in \Gamma^*$ and $q \in Q_G$, we have $w \notin F_G(q)$ if and
only if there is a path labeled $w$ from $q$ to $0$ in $G^0$. 

If $G$ and $H$ are labeled graphs, then the \term{label product graph
  of $G$ and $H$} is the graph $G*H$ whose
vertices are $Q_G \times Q_H$ and with an edge between $(p_1, p_2)$
and $(q_1, q_2)$ labeled $\ell$ if and only if there is an edge
labeled $\ell$ from $p_1$ to $q_1$ in $G$ and an edge labeled $\ell$
from $p_2$ to $q_2$ in $H$.  One can show that for
$w \in (\Ac_G \cup \Ac_H)^*$, there is a path labeled $w$ from
$p_1$ to $q_1$ in $G$ and a path labeled $w$ from $p_2$ to $q_2$ in
$G$ if and only if there is a path labeled $w$ from $(p_1, p_2)$ to
$(q_1, q_2)$ in $G*H$.

Let $G$ be a labeled graph, let $G^0$ be the
sink vertex graph of $G$ with alphabet $\Ac_G$ and let
$G^0 * G^0$ be the label product graph of $G^0$ and $G^0$.
With the properties of the auxillary graphs, one can show that
the following conditions are equivalent to cases (i)-(iii) from the
pair-synchronizing definition.

\begin{enumerate}[(I)]
\item there is a path in $G^0 * G^0$ from $(p,q)$ to $(r, 0)$ for some
  vertex $r$ in $G$;
\item there is a path in $G^0 * G^0$ from $(p,q)$ to $(0, r)$ for some
  vertex $r$ in $G$;
\item there is a path in $G^0 * G^0$ from $(p,q)$ to $(r, r)$ for some
  vertex $r$ in $G$.
\end{enumerate}

Using, say, a depth-first search, one can determine if there is
a pair-synchronizing word for a given pair of vertices by testing for
the existence of a path satisfying one of (I)-(III).
The size of
$G^0 * G^0$ is $O(|Q_G|^2 \cdot |\Ac_G|)$, so one can
construct the graph and query the existence of such path in polynomial
time.
Each iteration of Algorithm
\ref{alg:sw} therefore takes polynomial time, and furthermore, as there are at most $|Q_G|$ iterations, in total, the algorithm will take
polynomial time.

\subsection{Testing for subshift}
\label{sec:testing-subshift}

We now turn to the \inclusionp problem for deterministic presentations $G$ and $H$ where $G$ is irreducible.
The key idea behind the algorithm is to try to find a word exhibiting the fact that $\shift{G} \nsubseteq \shift{H}$.
We say that $w$ \term{separates $G$ from $H$} if
$Q_G \cdot w \neq \varnothing$ while $Q_H \cdot w = \varnothing$.
When $G$ and $H$ are essential, the existence of such a word is
equivalent to $\shift{G} \nsubseteq \shift{H}$.
Algorithm~\ref{alg:subshift} adapts the algorithm for synchronizing words to find such separating words, thus showing \inclusionp for $G$ and $H$ is in \cc{P} when $G$ is irreducible.
We show in Theorem~\ref{thm:inclusion-hard} that the general problem is \cc{PSPACE}-complete.

We state the correctness of Algorithm~\ref{alg:subshift} for the more general case of labeled graphs, which need not be essential, since we rely on that case for Theorem~\ref{thm:sync-testing-alg}.
\begin{theorem}\label{thm:subshift-correct}
  Given deterministic labeled graphs $G$ and $H$, where $G$ is irreducible,
  Algorithm~\ref{alg:subshift}
  returns a word separating $G$ from $H$ if one exists, and returns
  nil otherwise.
\end{theorem}

Like in Algorithm~\ref{alg:sw}, Algorithm~\ref{alg:subshift} operates
by iteratively building a word $u$.
In addition, the algorithm fixes a vertex $p_0 \in Q_G$, and maintains a vertex $p\in Q_G$ and subset $X\subseteq Q_H$ satisfying the invariants
$u \in F_G(p_0)$, $p_0 \cdot u = p$, and $Q_H \cdot u = X$.
In each iteration of the main loop, the algorithm searches for a word $w$ such
that $w \in F_G(p)$ while $w \notin F_H(q)$ for some $q \in X$.
If one is found, the algorithm updates $u$ to $uw$, $p$ to $p \cdot w$, and
$X$ to $X \cdot w$, which maintains the invariants.
As
$w \notin F_H(q)$ and $q \in X$, we have $|X| > |X \cdot w|$, so
the algorithm again terminates in at most $|Q_H|$ iterations.

\begin{proof}[Proof of Theorem~\ref{thm:subshift-correct}]
  If Algorithm~\ref{alg:subshift} returns a non-nil value, it must have exited at line~\ref{alg:subshift:line:ret-nonnil}, so $|X| = 0$.
  The invariants give
  $Q_H \cdot u = X = \varnothing$ and $u \in F_G(p_0)$,
  meaning $Q_H \cdot u = \varnothing$ and $Q_G \cdot u \neq \varnothing$.
  Thus, $u$ separates $G$ from $H$.

  Conversely, if Algorithm~\ref{alg:subshift} returns nil, it must have exited at
  line~\ref{alg:subshift:line:ret-nil}, which implies there
  exist $p \in Q_G$ and $q \in Q_H$ such that there is no word
  $w$ with $w \in F_G(p)$ and $w \notin F_H(q)$.
  We show below that, if some word separates $G$ from $H$,
  then for every $p' \in Q_G$ and $q' \in Q_H$, there is a word
  $w$ with $w \in F_G(p')$ and $w \notin F_H(q')$.
  By contraposition, therefore, no word separates $G$ from $H$.

  Suppose there is a word $w$ separating $G$ from $H$, so
  that we have $Q_G \cdot w \neq \varnothing$ while
  $Q_H \cdot w = \varnothing$.
  Then, there is some vertex
  $p^* \in Q_G$ such that $w \in F_G(p^*)$ and $w \notin F_H(q')$ for
  every $q' \in Q_H$.
  Let $p \in Q_G$ and $q \in Q_H$. As $G$ is
  irreducible, there is some $u \in F_G(p)$ such that
  $p \cdot u = p^*$,
  giving $uw \in F_G(p)$.
  Let $q' \triangleq q \cdot u$.
  By the above, $w \notin F_H(q')$, and thus $uw \notin F_H(q)$.
\end{proof}

\begin{algorithm}
  \caption{Subshift testing}
  \label{alg:subshift}
  \begin{algorithmic}[1]
    \Require $G$ is an irreducible deterministic labeled graph
    \Require $H$ is a deterministic labeled graph
    \Procedure{separating-word}{$G$, $H$}
    \State  $p_0 \gets$ any element in $Q_G$;
    \quad $p \gets p_0$;\quad $X \gets Q_H$;\quad $u \gets \epsilon$
    \While{$|X| > 0$}
    \State $q \gets$ any element in $X$
    \State find a word $w$ such that $w \in F_G(p)$ and $w \notin F_H(q)$ 
    \If{$w$ exists}
    \State $p \gets p \cdot w$;\quad $X \gets X \cdot w$;\quad $u \gets uw$
    \Else
    \State \Return nil \label{alg:subshift:line:ret-nil}
    \EndIf
    \EndWhile
    \State \Return u \label{alg:subshift:line:ret-nonnil}
    \EndProcedure
  \end{algorithmic}
\end{algorithm}

Analagously to Algorithm~\ref{alg:sw}, we can implement Algorithm
\ref{alg:subshift} in polynomial time by noticing that the existence
of a word $w$ such that $w \in F_G(p)$ and $w \notin F_H(p)$ is
equivalent to the existence of a path in $G * H^0$ from $(p,q)$ to
$(r, 0)$ for some vertex $r$ in $G$, where $H^0$ is the sink vertex
graph of $H$ with alphabet $\Ac_G \cup \Ac_H$ and $G * H^0$ is
the label product graph of $G$ and $H^0$.

\subsection{Testing for synchronizing presentations}
\label{sec:test-synchr-pres}

With Algorithm~\ref{alg:sw} and Algorithm~\ref{alg:subshift}, we can now
establish a polynomial-time algorithm for checking if a given
deterministic graph is synchronizing, given by Algorithm~\ref{alg:sync-test}. The correctness of the algorithm is implied by
the following characterization of a synchronizing presentation.

\begin{theorem}
  \label{thm:sync-testing-alg}
  Let $G$ be a deterministic labeled graph with vertex set $Q$.
  Then, $G$ is synchronizing if and only if for each initial irreducible
  component $C$, there exists (i) a synchronizing word for the subgraph induced
  by $C$ and (ii) a word separating
  the subgraph induced by $C$ from
  the subgraph induced by $Q \setminus C$.
\end{theorem}

\begin{proof}[]
  Suppose $G$ is synchronizing.
  Let $C$ be an initial component of $G$, and fix $r \in C$.
  As $G$ is synchronizing, let $w$ be a word that synchronizes to $r$ in $G$.
  As $w$ is synchronizing for $G$, there is some vertex $p \in Q$ such
  that $p \cdot w = r$.
  Since $r \in C$ and $C$ is initial, we must have $p \in C$.
  Thus $C \cdot w = \{r\}$, establishing (i).
  As $C$ is initial and $r\in C$, we cannot have $q \cdot w = r$ for any $q \notin C$.
  We conclude $(Q \setminus C) \cdot w = \varnothing$.
  As $C \cdot w = \{r\} \neq \varnothing$, we have (ii).

  Conversely, suppose for each initial irreducible
  component $C$, (i) there is a synchronizing word $u_C$ for the subgraph induced
  by $C$ and (ii) there is a word $w_C$ separating the subgraph induced by $C$ from the subgraph induced by
  $Q \setminus C$.
  Let $r$ be any vertex in $G$.
  Let $C$ be an initial irreducible component such that $r$ is reachable
  from every vertex in $C$.
  Condition (i) gives $C \cdot u_C = \{p\}$ for some $p \in C$.
  Condition (ii) gives some vertex $q \in C$ with
  $w_C \in F_G(q)$ such that $q \cdot w_C \in C$ and $(Q \setminus C) \cdot w_C = \varnothing$.
  As $C$ is an irreducible component and $p, q \in C$, then there is
  some word $x$ such that $p \cdot x = q$. As $q \cdot w_C \in C$ and $r$ is
  reachable from every vertex in $C$, there is
  some word $y$ such that $(q \cdot w_C) \cdot y = r$.
  Combining the above with a straightforward calculation for $Q\setminus C$, we have
  \begin{align*}
    C \cdot u_C x w_C y &= \{p\} \cdot x w_C y = \{q\} \cdot w_C y = \{q \cdot w_C\} \cdot y = \{r\},
    \\
    (Q \setminus C) \cdot u_C x w_C y &= ((Q \setminus C) \cdot u_C x) \cdot w_C y \subseteq (Q \setminus C) \cdot w_C y = \varnothing \cdot y = \varnothing.
  \end{align*}
  Thus, $Q \cdot u_C x w_C y = (C \cdot u_C x w_C y) \cup ((Q\setminus C) \cdot u_C x w_C y) = \{r\}$.
  As $r$ was arbitrary, $G$ is synchronizing.
\end{proof}

\begin{algorithm}
  \caption{Recognizing synchronizing presentations}
  \label{alg:sync-test}
  \begin{algorithmic}[1]
    \Require $G$ is a deterministic labeled graph
    \Procedure{is-synchronizing}{$G$}
    \State $\mathcal{C} \gets$ initial irreducible components of $G$
    \For{$C \in \mathcal{C}$}
    \State $G[C] \gets$ subgraph induced by $C$
    \State $G[\overline{C}] \gets$ subgraph induced by $Q_G \setminus C$
    \State $u \gets$ \Call{synchronizing-word}{$G[C]$}
    \State $v \gets$ \Call{separating-word}{$G[C]$, $G[\overline{C}]$}\label{alg:is-sync:line:essential}
    \If{$u$ is nil or $v$ is nil}
    \State \Return false
    \EndIf
    \EndFor
    \State \Return true
    \EndProcedure
  \end{algorithmic}
\end{algorithm}

\subsection{SFT testing for synchronizing deterministic presentations}
\label{sec:sft-test-synchr}

The proof of Theorem 3.4.17 of \citet{lind2021introduction} implicitly describes a polynomial-time algorithm to test whether an irreducible sofic shift, given as an irreducible deterministic presentation, is an SFT, and \citet{schrock2019complexity} gives a similar algorithm explicitly.
We extend these algorithms to synchronizing deterministic presentations.

For a deterministic labeled graph $G$, we define the labeled graph $\hat{G}$ as the label product graph $G * G$
(see Section~\ref{sec:find-synchr-words}) with the diagonal vertices removed, i.e., those of the form $(q,q)$.
Given a follower-separated synchronizing deterministic presentation $G$,
the algorithm to recognize if $\shift{G}$ is an SFT is to simply test if
the graph $\hat{G}$ is acyclic.
(A \term{cycle} is a nonempty path
that starts and ends at the same vertex, and we say a graph is \term{acyclic} if it has no cycle.)
This algorithm runs in polynomial time, as the size of $\hat{G}$ is quadratic with respect to the size of $G$, and it is well-known that one can test whether a directed graph is acyclic in linear time.

To show the correctness of this algorithm, we first show that
$\hat{G}$ characterizes the nonsynchronizing words of $G$.

\begin{lemma}\label{lem:ghat-char}
  Let $G$ be a deterministic presentation and let $w \in \Bc(\shift{G})$.
  Then, there is a path in $\hat{G}$ labeled $w$
  if and only if
  $w$ is not synchronizing for $G$.
\end{lemma}

\begin{proof}
  Suppose there is a path $\pi$ in $\hat{G}$ labeled $w$, from $(p,q)$ to $(p',q')$.
  Then, we have $p \neq q$ and $p' \neq q'$, and $p \cdot w = p'$ and $q \cdot w = q'$.
  Thus, $Q_G \cdot w \supseteq \{p, q\} \cdot w = \{p', q'\}$.
  As $p' \neq q'$, we have $|Q_G \cdot w| \geq 2$ so $w$ is not synchronizing for $G$.
  Conversely, if $w$ was not synchronizing for $G$, then $|Q_G \cdot w| \geq 2$.
  Let $p',q' \in Q_G \cdot w$ be distinct.
  Let $p,q \in Q_G$ such that $p \cdot w = p'$ and $q \cdot w = q'$.
  If for some factoring $w = uv$ we had $p \cdot u = q \cdot u$,
  then $p' = p \cdot w = (p \cdot u) \cdot v = (q \cdot u) \cdot v = q \cdot w = q'$, a contradiction to $p'$ and $q'$ being distinct.
  Thus, there is a path labeled $w$ in $G*G$ from $(p,q)$ to $(p',q')$ which does not pass through any diagonal vertices, meaning it is a labeled path in $\hat G$.
\end{proof}

Because of the correspondence of synchronizing and intrinsically
synchronizing words in follower-separated synchronizing deterministic
presentations, we can use $\hat{G}$ to characterize when $\shift{G}$
is an SFT. 

\begin{theorem}
  \label{thm:sft-testing-sdp}
  Let $G$ be a follower-separated synchronizing deterministic
  presentation.  Then, $\shift{G}$ is an SFT if and only if
  $\hat{G}$ is acyclic.
\end{theorem}

\begin{proof}
  Suppose $\hat{G}$ had a cycle.
  By Theorem~\ref{thm:sft-char},
  to show that $\shift{G}$ is not an SFT,
  it suffices to show that
  for every $M \geq 0$, there exists a word $w \in \Bc(\shift{G})$ with $|w| \geq M$ that is not intrinsically synchronizing for $\shift{G}$.
  Let $M \geq 0$.
  Since $\hat{G}$ has a cycle, in particular it has a path of any length.
  Let $\pi$ be a path in $\hat{G}$ of length at least $M$, and let $w$ be its label.
  We have $w \in \Bc(\shift{G})$ and $|w| \geq M$.
  By Lemma~\ref{lem:ghat-char}, $w$ is not synchronizing for $G$, and
  by Lemma~\ref{lem:sync-word-correspondence}, $w$ is therefore not
  intrinsically synchronizing for $\shift{G}$. 
  
  Conversely, suppose $\hat{G}$ is acyclic, and let $M \triangleq |Q_{\hat{G}}|$.
  By Theorem~\ref{thm:sft-char}, to show $\shift{G}$ is an SFT,
  it suffices to show that every word $w \in \Bc(\shift{G})$ with $|w| \geq M$ is intrinsically synchronizing for $\shift{G}$.
  Let $w \in \Bc(\shift{G})$ and suppose $|w| \geq M$.
  Suppose for a contradiction that $w$ is not intrinsically synchronizing for
  $\shift{G}$.
  By Lemmas~\ref{lem:sync-word-correspondence} and~\ref{lem:ghat-char} once again,
  there is a path in $\hat{G}$ labeled $w$, of length $|w| \geq M$.
  As $G$ is acyclic, however, every path in $\hat{G}$ must have length strictly less than $|Q_{\hat{G}}| = M$, a contradiction.
  Thus, $w$ must be intrinsically synchronizing for $\shift{G}$.
\end{proof}

\begin{remark}\label{re:sft-bound}
  Let $X$ be a shift space and $M \geq 0$. Say $X$ is
  \term{M-step} if every word $w \in \Bc(X)$ with $w \geq M$ is
  intrinsically synchronizing for $X$. With this definition and
  rephrasing Theorem~\ref{thm:sft-char}, a shift space is an SFT if and
  only if it is $M$-step for some $M \geq 0$. The converse direction
  then implies that if $\hat{G}$ is acyclic, then it must be
  $(|Q_G|^2 - |Q_G|)$-step, as $|Q_{\hat{G}}| = |Q_G|^2 - |Q_G|$.
  Thus, if $G$ is a follower-separated synchronizing deterministic
  presentation, then $\shift{G}$ is an SFT if and only if it is
  $(|Q_G|^2 - |Q_G|)$-step.
  (Cf.~\cite[Theorem 3.4.17]{lind2021introduction}.)
\end{remark}

\subsection{Isomorphism and equality}
\label{sec:iso-eq}

A \term{homomorphism} between deterministic labeled graphs $G$ and $H$
is a mapping $\varphi \colon Q_G \to Q_H$ that preserves the
transition action: we have $F_G(q) = F_H(\varphi(q))$
and $\varphi(q \cdot w) = \varphi(q) \cdot w$
for all $q \in Q_G$ and $w \in F_G(q)$. An
\term{isomorphism} is a bijective homomorphism.
In general, the problem of deciding isomorphism between deterministic labeled graphs is \cc{GI}-complete, meaning it has a polynomial-time many-one reduction to and from the graph isomorphism problem on unlabeled graphs \cite{booth1978isomorphism}.
For follower-separated graphs, however, the problem is in \cc{P}. 
To show this, we need the following lemma, which states that 
preserving the follower set of a vertex is sufficient for being
a homomorphism onto a follower-separated graph.

\begin{lemma}\label{lem:hom-follower-set}
  Let $G$ and $H$ be deterministic labeled graphs, and $\varphi \colon Q_G \to Q_H$
  a map between their vertices. If $H$ is follower-separated, then
  $\varphi$ is a homomorphism if and only if $F_G(q) = F_H(\varphi(q))$ for
  all $q \in Q_G$.
\end{lemma}

\begin{proof}
  That homomorphisms preserve follower sets follows directly from the definition.
  For the converse, suppose $F_G(q) = F_H(\varphi(q))$ for all $q \in Q_G$.
  Let $q \in Q_G$ and $w \in F_G(q)$. As $H$ is follower-separated, 
  it suffices to show that
  $F_H(\varphi(q \cdot w)) = F_H(\varphi(q) \cdot w)$ to show that
  $\varphi(q \cdot w) = \varphi(q) \cdot w$. For any $u$, we have 
  \begin{align*}
    &\phantom{{}\iff{}} u \in F_H(\varphi(q \cdot w)) \\
    &\iff u \in F_G(q \cdot w) \\
    &\iff wu \in F_G(q) \\
    &\iff wu \in F_H(\varphi(q)) \\
    &\iff u \in F_H(\varphi(q) \cdot w).
  \end{align*}
  Thus $F_H(\varphi(q \cdot w)) = F_H(\varphi(q) \cdot w)$.
\end{proof}

Now, given two follower-separated deterministic labeled graphs $G$ and $H$,
we can test if they are isomorphic by taking the disjoint union graph $G+H$, computing the follower-equivalences of $G+H$, and testing
if all the follower-equivalence classes are pairs (i.e. sets of size $2$).
As $G$ and $H$ are follower-separated, if two distinct vertices in $G+H$ are
follower-equivalent, then one of them must be a vertex from $G$
and the other from $H$. Thus, if all the follower-equivalence classes
are pairs, then a bijective map $\varphi \colon Q_G \to Q_H$ that preserves
the follower set of a vertex can be read off from the pairs. By
Lemma~\ref{lem:hom-follower-set}, this map is an isomorphism.
Conversely, if $\varphi \colon Q_G \to Q_H$ is an isomorphism,
then for $p \in Q_G$ and $q \in Q_{H}$ with $F_G(p) = F_H(q)$,
then $F_H(\varphi(p)) = F_H(q)$ and so $\varphi(p) = q$. In other words,
for any $p \in Q_G$ and any $q \in Q_H$ follower-equivalent
to $p$, then $\varphi(p) = q$. This implies that
all the follower-equivalence classes of $G+H$ are pairs.

Since the follower set of a vertex is preserved under an isomorphism,
if $G$ and $H$ are isomorphic deterministic presentations, then
$\shift{G} = \shift{H}$. However, even for follower-separated
presentations, the converse is not necessarily true. (See Figure \ref{fig:example}.)
But \citet[Corollary 5.4] {jonoska1996sofic}
proved that any two follower-separated synchronizing deterministic
presentations of the same sofic shift are isomorphic, which
implies that recognizing isomorphism is sufficient for
recognizing if $\shift{G} = \shift{H}$ when $G$ and $H$ are
follower-separated synchronizing deterministic presentations.
Furthermore, this implies that \equalp is in \cc{P} for
synchronizing deterministic presentations, as given
synchronizing deterministic presentations $G$ and $H$,
to determine if $\shift{G} = \shift{H}$,
one can test if $G/{\sim}$ and $H/{\sim}$, the
follower-separations of $G$ and $H$, are isomorphic.

\begin{figure}
  \centering
  \begin{tikzpicture}
    \node[vertex] (q1) at (0,0) {$q_1$};
    \node[vertex] (q2) at (2,0) {$q_2$};
    \node[vertex] (q3) at (4,0) {$q_3$};

    \draw (q1) to node[above] {$1$} (q2);
    \draw (q1.south west) to[left loop] node[left] {$0$} (q1.north west);
    \draw (q2) to[out=20, in=180-20] node[above] {$0$} (q3);
    \draw (q2.north west) to[left loop] node[above] {$1$} (q2.north east);
    \draw (q3) to[out=180+20, in=0-20] node[below] {$0$} (q2);
  \end{tikzpicture}
  \caption{A reducible, follower-separated, deterministic presentation $G$.
    Let $H$ be the subgraph induced by $q_2$ and $q_3$, which is irreducible and follower-separated.
    Then, $\shift{G} = \shift{H}$, but $G$ and $H$ are not isomorphic.} \label{fig:example}
\end{figure}
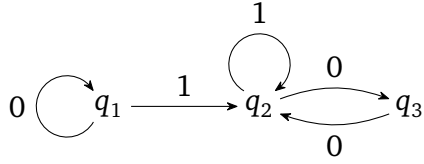

\section{Complexity Lower Bounds}
\label{sec:complexity-lower-bounds}

In the following sections, we show all the decision problems in Table
\ref{tab:problems} are \cc{PSPACE}-hard. In Appendix~\ref{apx:pspace}, we show all
those decision problems are in \cc{PSPACE}. As a result, we have the
following.

\begin{theorem}\label{thm:hardness-all}
  Every problem in Table~\ref{tab:problems} is \cc{PSPACE}-complete.
\end{theorem}

To establish the hardness of these decision problems, we will leverage hardness results from the automata theory literature.
To relate automata to sofic shifts, we will treat automata as a type of labeled graph.
Formally, we define a \term{deterministic finite automaton (DFA)} to be a fully deterministic labeled graph $M$ with a designated initial state $s \in Q_M$ and set of accepting states $F \subseteq Q_M$.
For DFAs, following convention from the automata literature, we
will write the transition action as a function $\delta(q, w) \triangleq q \cdot w$.
The \term{language} of $M$ is the set
$L(M) \triangleq \{\, w \in \Ac_M^* : \delta(s, w) \in F
\,\}$.
Note that $L(M)$ may differ from $\Bc(\shift{M})$, the language of the sofic shift presented by $M$.
In fact, as DFAs are fully deterministic, we always have $\Bc(\shift{M}) = \Ac_M^*$, meaning $\shift{M}$ is always the full shift.

We will reduce from the \term{DFA intersection nonemptiness problem}
(\dfacapp) and \term{DFA union universality problem} (\dfacupp), both
of which are \cc{PSPACE}-complete.
The \dfacapp problem asks whether, given $n$ DFAs $M_1, \dots, M_n$ over a common input alphabet $\Sigma$, is $\bigcap_{i=1}^n L(M_i) \neq \varnothing$?
Similarly, the \dfacupp problem asks whether, given $n$ DFAs $M_1, \dots, M_n$ over a common input alphabet $\Sigma$, is $\bigcup_{i=1}^n L(M_i) = \Sigma^*$?
\citet{kozen1977lower} showed that \dfacapp is \cc{PSPACE}-complete; one can see that \dfacupp is \cc{PSPACE}-complete from the following two facts: (i) the complement of \dfacapp is \cc{PSPACE}-complete, and (ii) $\bigcap_{i=1}^n L(M_i) = \varnothing$ if and only if $\bigcup_{i=1}^n L(\overline{M_i}) = \Sigma^*$, where $\overline{M_i}$ is $M_i$ with the accepting states being the complement of the accepting states of $M_i$.
Within our reductions, for an instance $M_1, \dots, M_n$ of \dfacupp or \dfacapp, we will let $Q_i$, $\delta_i$, $s_i$, and $F_i$ denote the set of states, transition function, initial state, and set of accepting states for $M_i$.
(The $Q_i$ are assumed to be pairwise disjoint.)

\subsection{Hardness of equality, containment, irreducibility, and SDP
  existence}
\label{sec:hardn-equal-cont}

In this section, we give a single polynomial-time reduction, which reduces \dfacupp simultaneously to \inclusionp, \equalp, \irredp, and \sdpp, giving the following.

\begin{theorem}\label{thm:hardness-four}
  \inclusionp, \equalp, \irredp, and \sdpp are \cc{PSPACE}-hard.
\end{theorem}

The idea behind the reduction is to create pre-initial states $p_i$ for each DFA $M_i$, and chain these together in a loop, with special symbols $\leftm$ into and $\rightm$ out of each DFA.
We then add a special state $p^*$ in its own initial irreducible component, whose follower set contains $\{\,\leftm w \rightm : w\in\Sigma^*\,\}$.
(See Figure~\ref{fig:irred-reduction} for a visualization.)
Letting $H$ be the whole graph minus the special state $p^*$, we can therefore test whether the DFA languages union to $\Sigma^*$ by asking whether $\shift{G} = \shift{H}$, i.e., whether $p^*$ was needed to cover all possible strings $w\in\Sigma^*$ between $\leftm$ and $\rightm$.
Equivalently, we could test $\shift{G} \subseteq \shift{H}$, since the reverse inclusion is immediate.
As $H$ is an irreducible presentation, we could also test whether $\shift{G}$ is irreducible.
Finally, the reduction to \sdpp follows for the following reasons:
first, $H$ is a synchronizing determinstic presentation (as
$\rightm \ell^{i-1}$ synchronizes to $p_i$ for all $i$), so when the
langauges of the DFAs union to $\Sigma^*$, $H$ is a synchronizing
deterministic presentation for $\shift{G}$; second, when there is a
word $w \in \Sigma^*$ not in the language of any of the DFAs, one can
show $\shift{G}$ does not have a synchronizing deterministic
presentation by invoking Theorem \ref{thm:sdp-char} and showing that any $u$ such that
$u \leftm w \rightm \in \Bc(\shift{G})$ is not intrinisically
synchronizing.

\begin{figure}[t]
  \centering
  {
    \begin{tikzpicture}[rounded corners=10pt]

      \path (1.5*\xgap , 2.25*\ygap) node[vertex] (dots) {$\vdots$};
      \path (0, 2*\ygap) node[vertex] (Pdots) {$\vdots$};
      \path (0, 4*\ygap) node[vertex] (P*) {$p^*$} ++(1.5*\xgap,0) node[vertex] (s*) {$s^*$};

      \fill (0, 0) node[vertex] (P1) {$p_1$} 
      ++(\xgap , 0) node[vertex] (s1) {$s_1$} 
      +(0.5*\xgap, 0.35*\ygap) node[vertex] (M1) {$M_1$}
      +(0.5*\xgap,0) node {$\dots$}
      ++(\xgap , 0) circle[radius=1pt] node[draw, circle, inner sep=2pt] (f1) {};

      \fill (0, \ygap) node[vertex] (P2) {$p_2$} 
      ++(\xgap , 0) node[vertex] (s2) {$s_2$}
      +(0.5*\xgap, 0.35*\ygap) node[vertex] (M2) {$M_2$}
      +(0.5*\xgap,0) node {$\dots$}
      ++(\xgap , 0) circle[radius=1pt] node[draw, circle, inner sep=2pt] (f2) {};

      \fill (0, 3*\ygap) node[vertex] (Pn) {$p_n$} 
      ++(\xgap , 0) node[vertex] (sn) {$s_n$}
      +(0.5*\xgap, 0.35*\ygap) node[vertex] (Mn) {$M_n$}
      +(0.5*\xgap,0) node {$\dots$}
      ++(\xgap , 0) circle[radius=1pt] node[draw, circle, inner sep=2pt] (fn) {};

      \def\ingap{\arrowgap}

      \begin{scope}[-{Stealth[round,width=3pt]}]
        \draw (f1) -| (3*\xgap, -\ygap) -| ($(P1.south)+(1.5*\ingap, 0)$);
        \draw (f2) -| (3*\xgap+\arrowgap, -\ygap-\arrowgap) -| ($(P1.south)+(.5*\ingap, 0)$);
        \draw (fn) -| (3*\xgap+2*\arrowgap, -\ygap-2*\arrowgap) -| ($(P1.south)-(.5*\ingap, 0)$);
        \draw (s*) -| (3*\xgap+3*\arrowgap, -\ygap-3*\arrowgap) -| ($(P1.south)-(1.5*\ingap, 0)$);
      \end{scope}

      \draw (P1) to node[left] {$\ell$} (P2);
      \draw (P2) to node[left] {$\ell$} (Pdots);
      \draw (Pdots) to node[left] {$\ell$} (Pn);

      \draw (P1) to node[above] {$\leftm$} (s1);
      \draw (P2) to node[above] {$\leftm$} (s2);
      \draw (Pn) to node[above] {$\leftm$} (sn);
      \draw (P*) to (s*);
      \path (P*) ++(\xgap, 0) node (phantom) {$\phantom{s_n}$};
      \path (P*) to node[above] {$\leftm$} (phantom);

      \path (f1) to node[above] {$\rightm$} +(\xgap, 0);
      \path (f2) to node[above] {$\rightm$} +(\xgap, 0);
      \path (fn) to node[above] {$\rightm$} +(\xgap, 0);
      \path (P*) ++(2*\xgap, 0) to node[above] {$\rightm$} +(\xgap, 0);

      \node[draw, dashed, rectangle, fit=(M1) (s1) (f1)] {};
      \node[draw, dashed, rectangle, fit=(M2) (s2) (f2)] {};
      \node[draw, dashed, rectangle, fit=(Mn) (sn) (fn)] {};

      \draw (P1.south west) to[left loop] node[left] {$*$} (P1.north west);
      \draw (P2.south west) to[left loop] node[left] {$*$} (P2.north west);
      \draw (Pn.south west) to[left loop] node[left] {$*$} (Pn.north west);
      \draw (P*.south west) to[left loop] node[left] {$*$} (P*.north west);
      \draw (s*.north west) to[left loop] node[above] {$\Sigma$} (s*.north east);

      \draw (Pn) to[out=10, in=180-10, relative] node[above left] {$\ell$} (s*);
    \end{tikzpicture}
  }
  \caption{Schematic of Reduction~\ref{red:irred}. The edges labeled $\ell$
  from each state $q \in Q_i$ to $s_i$ are not pictured.}
  \label{fig:irred-reduction}
\end{figure}
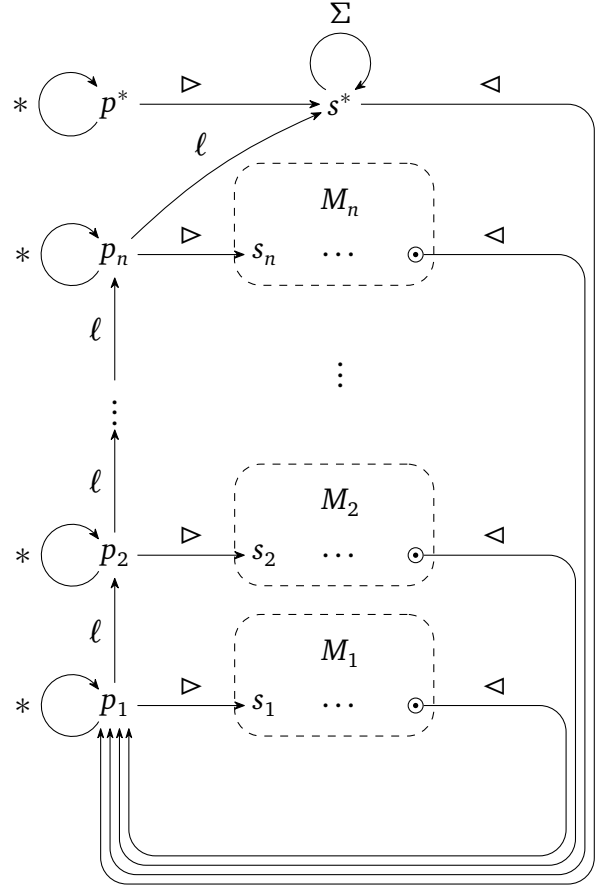

\begin{reduction}\label{red:irred}
  Let $M_1, \dots, M_n$ be an instance to the \dfacupp problem.
  Construct the deterministic presentation $G$ as follows.
  For each $i=1, \dots, n$,
  \begin{enumerate}
  \item add a state $p_i$ (the $i$th \term{pre-initial} state) to $G$;
  \item embed $M_i$ into $G$;
  \item add a self loop labeled $*$ on $p_i$;
  \item add an edge labeled $\leftm$ from $p_i$ to the corresponding
    initial state $s_i$;
  \item for each accepting state $q \in F_i$, add an edge labeled
    $\rightm$ from $q$ to $p_1$;
  \item for each state $q \in Q_i$, add an edge labeled $\ell$ from $q$
    to $s_i$.

  \end{enumerate}
  Then, add two states $p^*$ and $s^*$, add a self loop labeled $*$
  on $p^*$, add an edge labeled $\leftm$ from $p^*$ to $s^*$,
  and add an edge labeled $\rightm$ from $s^*$ to $p_1$.
  For each $a \in \Sigma$, add a self loop labeled $a$ on $s^*$.
  For each $i=1,\dots,n-1,$ add an edge labeled $\ell$ from $p_i$
  to $p_{i+1}$. Finally, add an edge labeled $\ell$ from $p_n$ to $s^*$.
  (See Figure~\ref{fig:irred-reduction}.)
\end{reduction}

Let $G$ be the deterministic presentation obtained from Reduction~\ref{red:irred}
on an instance $M_1, \dots, M_n$. 
Without loss of generality, for each $i$, we may assume that
(I) $F_i \neq \varnothing$, as otherwise, if $F_i = \varnothing$,
then $L(M_i) = \varnothing$ so thus $L(M_i)$ does not contribute to the union;
and (II) every state $q \in Q_i$ is reachable from $s_i$, as when
modifying $M_i$ to $M_i'$ by removing those unreachable states, we have
$L(M_i) = L(M_i')$.
Let $H$ be the subgraph in $G$ induced by every vertex but $p^*$.
The following lemma summarizes several useful properties of the reduction.

\begin{lemma}\label{lem:irred-reduction-properties}
  The following hold of Reduction~\ref{red:irred}.
  \begin{enumerate}[(i)]
  \item $H$ is synchronizing and irreducible, and $G$ is essential;
  \item $\Sigma^* \subseteq F_G(q)$ for all $q \in \bigcup_{i=1}^n Q_i$;
  \item $\leftm w \rightm \in F(p_i)$ if and only if $w \in L(M_i)$;
  \item $\leftm w \rightm \in \Bc(\shift{H})$ if and only if $w \in \bigcup_{i=1}^n L(M_i)$.
  \item If $\bigcup_{i=1}^n L(M_i) \neq \Sigma^*$, there exists $w\in\Sigma^*$ with $\leftm w \rightm \in \Bc(\shift{G}) \setminus \Bc(\shift{H})$.
  \end{enumerate}
\end{lemma}
\begin{proof}
  For (i), by assumption (I), there exists a state in $Q_i$ with
  an edge labeled $\rightm$ to $p_1$. By assumption (II), every
  state is reachable from $s_i$, so there exists a path from $s_i$
  to $p_1$. Thus for any state in $Q_i$, one can always find a
  way to $p_1$ by returning to $s_i$ via an $\ell$ edge, and then
  finding a way to $p_1$. As $p_1$ can reach any other vertex in $H$, any state in $Q_i$ can reach any other vertex in $H$.
  From this, we can see that $H$ is irreducible, and it follows that $G$ is essential.
  We have that $H$ is synchronizing as $Q_G \cdot \rightm = \{p_1\}$ and every vertex in $H$ is reachable from $p_1$.

  For the other statements, first note that each of the $M_i$ are emulated
  by the transition action of $G$ in the following way:
  for $q \in Q_i$ and $w \in \Sigma^*$,
  we have $w \in F_G(q)$ and $q \cdot w = \delta_i(q, w)$
  and $q \in F_i$ if and only if $\rightm \in F_G(q)$.
  Thus, (ii) follows.
  For (iii), note that $p_i \cdot \leftm = s_i$ and
  $w \rightm \in F_G(s_i)$ if and only if $w \in L(M_i)$;
  thus, $\leftm w \rightm \in F_G(p_i)$ if and only if $w \in L(M_i)$.
  For (iv), note that $Q_H \cdot \leftm = \{s_1, \dots, s_n\}$; thus,
  by the previous observations, $\leftm w \rightm \in \Bc(\shift{H})$
  if and only if $w \in \bigcup_{i=1}^n L(M_i)$.
  Finally (v) follows from (iv) and the fact that $\leftm w \rightm \in \Bc(\shift{G})$ for all $w\in\Sigma^*$.
\end{proof}

With these properties, we can establish the correctness of Reduction~\ref{red:irred}.
The first theorem shows that it reduces \dfacupp to \inclusionp.

\begin{theorem}\label{thm:inclusion-hard}
     $\bigcup_{i=1}^n L(M_i) = \Sigma^*$ if and only if $\shift{G} \subseteq \shift{H}$.
\end{theorem}

\begin{proof}
  Suppose $\bigcup_{i=1}^n L(M_i) = \Sigma^*$.
  To establish $\shift{G} \subseteq \shift{H}$, we only need to show
  $F_G(p^*) \subseteq \Bc(\shift{H})$.
  Let $u \in F_G(p^*)$.
  If $p^* \cdot u = p^*$, then by construction, we have $u = *^m$ for some $m \geq 0$,
  and as $u \in F_G(p_1)$, then $u \in \Bc(\shift{H})$.
  Otherwise, if $p^* \cdot u = s^*$, then we can 
  factor $u$ into $u=*^m \leftm w$, where $m \geq 0$ and $w \in \Sigma^*$.
  Similarly, by Lemma~\ref{lem:irred-reduction-properties}(ii),
  we can find $u \in F_G(p_1)$, so $u \in \Bc(\shift{H})$.
  Finally, if $p^* \cdot u \notin \{p^*,s^*\}$, then we can factor 
  $u$ into $u = u_1 u_2$, where $u_1 = *^m \leftm w \rightm$ for some 
  $m \geq 0$ and $w \in \Sigma^*$, $p^* \cdot u_1 = p_1$, and $u_2 \in F_G(p_1)$.
  As $\bigcup_{i=1}^n L(M_i) = \Sigma^*$, Lemma~\ref{lem:irred-reduction-properties}(iii) implies $\leftm w \rightm \in \Bc(\shift{H})$, and in particular, there is some $i$ such that $\leftm w \rightm \in F_G(p_i)$.
  As $p_i \cdot \leftm w \rightm = p_i \cdot *^m \leftm w \rightm = p_i \cdot u_1 = p_1$,
  we have $u_1 \in F_G(p_i)$ and $u_2 \in F_G(p_i \cdot u_1)$.
  Thus $u_1 u_2 = u \in F_G(p_i)$, and $u \in \Bc(\shift{H})$.

  Conversely, Lemma~\ref{lem:irred-reduction-properties}(v) gives some $w\in\Sigma^*$ with $\leftm w \rightm \in \Bc(\shift{G}) \setminus \Bc(\shift{H})$.
  Hence, $\Bc(\shift{G}) \nsubseteq \Bc(\shift{H})$, and thus $\shift{G} \nsubseteq \shift{H}$.
\end{proof}

Immediately, as $\shift{H} \subseteq \shift{G}$, we have that
\dfacupp reduces to \equalp.

\begin{corollary}\label{cor:equal-correctness}
  $\bigcup_{i=1}^n L(M_i) = \Sigma^*$ if and only if $\shift{G} = \shift{H}$.
\end{corollary}

The reduction to \irredp now follows as well.

\begin{theorem}\label{thm:irred-hard}
  $\bigcup_{i=1}^n L(M_i) = \Sigma^*$ if and only if $\shift{G}$ is irreducible.
\end{theorem}

\begin{proof}
  If $\bigcup_{i=1}^n L(M_i) = \Sigma^*$, Corollary~\ref{cor:equal-correctness}
  implies $\shift{G} = \shift{H}$;
  as $\shift{H}$ is irreducible by Lemma~\ref{lem:irred-reduction-properties}(i), so is $\shift{G}$.
  Conversely, Lemma~\ref{lem:irred-reduction-properties}(v) gives some $w\in\Sigma^*$ with $\leftm w \rightm \in \Bc(\shift{G}) \setminus \Bc(\shift{H})$.
  As $\rightm$ synchronizes to $p_1$, for every $u \in F_{\shift{G}}(\leftm w \rightm)$, we have $F_{\shift{G}}(\leftm w \rightm u) = F_G(p_1 \cdot u)$.
  Furthermore, as $p_1 \cdot u \in Q_H$, we have $F_G(p_1\cdot u) \subseteq \Bc(\shift{H})$.
  Yet $\leftm w \rightm \notin \Bc(\shift{H})$, so we must have $\leftm w \rightm \notin F_G(p_1\cdot u)$.
  Thus, for every $u \in F_{\shift{G}}(\leftm w \rightm)$, we have
  $\leftm w \rightm \notin F_G(p_1\cdot u) = F_{\shift{G}}(\leftm w \rightm u)$.
  In other words, $\leftm w \rightm \in \Bc(\shift{G})$ but there is no word $u$
  such that $\leftm w \rightm u \leftm w \rightm \in \Bc(\shift{G})$.
\end{proof}

\begin{remark}
  Let $X$ be a shift space. Say $X$ is mixing if for every $u, v \in \Bc(X)$,
  there is an $N$ such that for every $n \geq N$, there is a word $w$
  with $|w| = n$ and $uwv \in \Bc(X)$. Mixing implies
  irreducibility, so if $\shift{G}$ is mixing, then $\shift{G}$ is
  irreducible. Note that $\shift{H}$ is mixing, as given $u, v \in \Bc(\shift{H})$,
  one can find words $w_1, w_2$ such that $uw_1$ synchronizes to $p_1$
  and $v \in F_G(p_1 \cdot w_2)$, and as $p_1 \cdot {*}^m = p_1$ for every $m$,
  we have that $u w_1 {*}^m w_2 v \in \Bc(\shift{H})$ for every $m$.
  Thus, $\shift{G}$ is irreducible if and only if it is mixing, so
  deciding if the sofic shift presented by a deterministic presentation
  is mixing is \cc{PSPACE}-hard.
  
  Similarly, say
  $X$ is nonwandering if for every $u \in \Bc(\shift{G})$, there
  is a word $w$ such that $uwu \in \Bc(X)$. Irreducibility implies
  nonwandering, so if $\shift{G}$ is irreducible, then $\shift{G}$
  is nonwandering. Note that the proof of Theorem \ref{thm:irred-hard} shows
  that if $\bigcup_{i=1}^n L(M_i) \neq \Sigma^*$, then
  $\shift{G}$ is not nonwandering. Thus, $\shift{G}$ is
  irreducible if and only if it is nonwandering, so
  deciding if the sofic shift presented by a deterministic presentation
  is nonwandering is \cc{PSPACE}-hard.
\end{remark}

Finally, we show that Reduction~\ref{red:irred} also reduces \dfacupp to \sdpp.
  
\begin{theorem}\label{thm:sdp-hard}
  $\bigcup_{i=1}^n L(M_i) = \Sigma^*$ if and only if $\shift{G}$ has a synchronizing
  deterministic presentation.
\end{theorem}

\begin{proof}
  If $\bigcup_{i=1}^n L(M_i) = \Sigma^*$, then Corollary~\ref{cor:equal-correctness} implies
  $\shift{G} = \shift{H}$, and as Lemma~\ref{lem:irred-reduction-properties}(i)
  implies $H$ is synchronizing, 
  then $\shift{G}$ has a synchronizing deterministic presentation.
  Conversely, Lemma~\ref{lem:irred-reduction-properties}(v) gives some $w\in\Sigma^*$ with $\leftm w \rightm \in \Bc(\shift{G}) \setminus \Bc(\shift{H})$.
  Suppose for a contradiction that $\shift{G}$ has a synchronizing deterministic presentation.
  By Theorem~\ref{thm:sdp-char}, there must be some $u \in \Bc(\shift{G})$
  such that $u$ is intrinsically synchronizing for $\shift{G}$ and
  $\leftm w \rightm \in F_{\shift{G}}(u)$.
  As $\leftm w \rightm \notin \Bc(\shift{H})$, the only vertex in $G$
  with $\leftm w \rightm$ in its follower set is $p^*$.
  Therefore, $\leftm w \rightm \in F_{\shift{G}}(u)$ implies $p^* \in Q_G \cdot u$.
  By construction, the only such $u$ take the form $u = *^k$ for some $k \geq 0$.
  However, $*^k$ is not intrinsically synchronizing for $\shift{G}$: we have
  $\rightm *^k \in \Bc(\shift{G})$ and
  $*^k \leftm w \rightm \in \Bc(\shift{G})$ but as
  $\leftm w \rightm \notin \Bc(\shift{H})$ and $\rightm *^k$
  synchronizes to a vertex in $H$, it must be the case that
  $\rightm *^k \leftm w \rightm \notin \Bc(\shift{G})$.
  Thus, $u=*^k$ is not intrinsically synchronizing, a contradiction.
  We conclude that $\shift{G}$ does not have a
  synchronizing deterministic presentation.
\end{proof}

\subsection{Hardness of SFT Testing and Minimization}
\label{sec:sft-reduct}

We now give a similar polynomial-time reduction, which reduces \dfacupp simultaneously to \sftp and \minimalp, giving the following.

\begin{theorem}\label{thm:sft-minimal-hard}
  The problems \sftp and \minimalp are \cc{PSPACE}-hard.
\end{theorem}

The reduction is similar in spirit to Reduction~\ref{red:irred}.
We still add edges labeled $\rightm$ out of each DFA into a terminal state, but
instead of adding edges labeled $\leftm$ into the DFAs from new pre-initial states, we
instead add these edges from within the DFAs to their corresponding initial states.
We also add self loops on each DFA state labeled $\ell$.
We then add a special state $s^*$ in its own initial component, whose follower set contains $\{w \rightm : w\in(\Sigma\cup\{\ell\})^*\}$.
See Figure~\ref{fig:sft-reduction} for a visualization.
The first observation we make is that, if and only if the DFA languages union to $\Sigma^*$, the shift $\shift{G}$ is presented by the graph $H$ in Figure~\ref{fig:sft-h-graph}.
Since $H$ presents an SFT, to show that we reduce to the \sftp problem, we need only argue that $\shift{G}$ is not an SFT when there is some word $w$ not in the language of any DFA.
Because the $\ell$ self loops arbitrarily delay the DFA decision to accept or reject, they prevent $\shift{G}$ from having a finite list of forbidden words, or equivalently, from being $M$-step for any finite $M$.
Finally, to show we reduce to \minimalp, we show that $\shift{G}$ does not have a $2$-vertex presentation when the DFA languages do not union to $\Sigma^*$.

\begin{figure}[t]
  \centering
  {
    \begin{tikzpicture}[rounded corners=10pt]

      \path (0.5*\xgap, 2.2*\ygap) node[vertex] (dots) {$\vdots$};

      \fill (0,0) node[vertex] (s1) {$s_1$} 
      +(0.5*\xgap, 0.35*\ygap) node[vertex] (M1) {$M_1$}
      +(0.5*\xgap,0) node {$\dots$}
      ++(\xgap , 0) circle[radius=1pt] node[draw, circle, inner sep=2pt] (f1) {};
      
      \fill (0, 1.2*\ygap) node[vertex] (s2) {$s_2$}
      +(0.5*\xgap, 0.35*\ygap) node[vertex] (M2) {$M_2$}
      +(0.5*\xgap,0) node {$\dots$}
      ++(\xgap , 0) circle[radius=1pt] node[draw, circle, inner sep=2pt] (f2) {};
      
      \fill (0, 3*\ygap) node[vertex] (sn) {$s_n$}
      +(0.5*\xgap, 0.35*\ygap) node[vertex] (Mn) {$M_n$}
      +(0.5*\xgap,0) node {$\dots$}
      ++(\xgap , 0) circle[radius=1pt] node[draw, circle, inner sep=2pt] (fn) {};

      \fill (0.5*\xgap, 4*\ygap) node[vertex] (s*) {$s^*$};

      \path (2*\xgap, 2*\ygap) node[vertex] (T) {$t$};
      \draw (T.south east) to[right loop] node[right] {$\star$} (T.north east);

      \path (f1) to node[above] {$\rightm$} +(\xgap, 0);
      \path (f2) to node[above] {$\rightm$} +(\xgap, 0);
      \path (fn) to node[above] {$\rightm$} +(\xgap, 0);
      \path (1*\xgap, 4*\ygap) node[circle, inner sep=2pt] (pf*) {}
      (pf*) to node[above] {$\rightm$} +(\xgap, 0);

      \draw (f1) -| ($(T.south) + (0.5*\arrowgap, 0)$);
      \draw (f2) -| ($(T.south) - (0.5*\arrowgap, 0)$);
      \draw (s*) -| ($(T.north) + (0.5*\arrowgap, 0)$);
      \draw (fn) -| ($(T.north) - (0.5*\arrowgap, 0)$);

      \draw (f1) to[bend left=20] node[below] (l1) {$\leftm$} (s1);
      \draw (f2) to[bend left=20] node[below] (l2) {$\leftm$} (s2);
      \draw (fn) to[bend left=20] node[below] (l3) {$\leftm$} (sn);

      \node[draw, dashed, rectangle, inner xsep=10pt, fit=(M1) (s1) (f1) (l1)] {};
      \node[draw, dashed, rectangle, inner xsep=10pt, fit=(M2) (s2) (f2) (l2)] {};
      \node[draw, dashed, rectangle, inner xsep=10pt, fit=(Mn) (sn) (fn) (l3)] {};

      \draw (s*.north west) to[left loop] node[above] {$\Sigma, \ell$} (s*.north east);

      \begin{scope}[left loop, arc center offset=3pt]
        \draw (s1.north west) to node[above] {$\leftm, \ell$} (s1.north east);
        \draw (s2.north west) to node[above] {$\leftm, \ell$} (s2.north east);
        \draw (sn.north west) to node[above] {$\leftm, \ell$} (sn.north east);

        \path (f1) node[vertex] (f1) {$\phantom{s_1}$};
        \path (f2) node[vertex] (f2) {$\phantom{s_1}$};
        \path (fn) node[vertex] (fn) {$\phantom{s_1}$};

        \draw (f1.north west) to node[above] {$\ell$} (f1.north east);
        \draw (f2.north west) to node[above] {$\ell$} (f2.north east);
        \draw (fn.north west) to node[above] {$\ell$} (fn.north east);
      \end{scope}

    \end{tikzpicture}
  }
  \caption{Schematic of Reduction~\ref{red:sft}.}
  \label{fig:sft-reduction}
\end{figure}
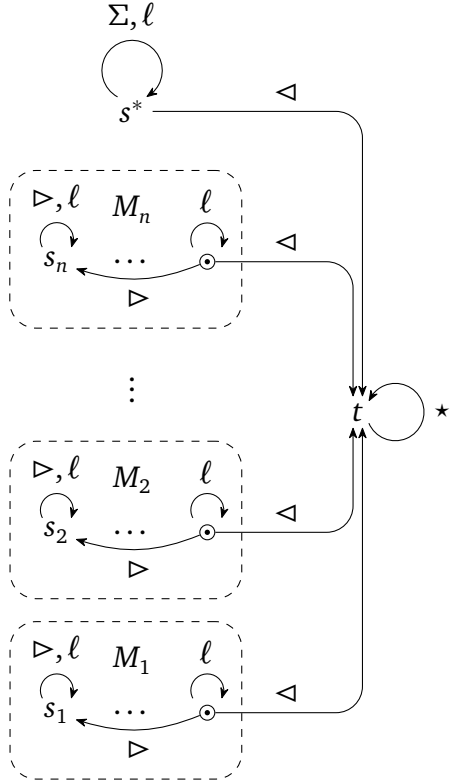

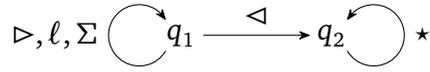
\begin{figure}[t]
  \centering
  \begin{tikzpicture}
    \node[vertex] (A) at (0,0) {$q_1$};
    \node[vertex] (B) at (2,0) {$q_2$};

    \draw (A.south west) to[left loop] node[left] {$\leftm, \ell, \Sigma$} (A.north west);
    \draw (B.south east) to[right loop] node[right] {$\star$} (B.north east);
    \draw (A) to node[above] {$\rightm$} (B);
  \end{tikzpicture}  
  \caption{The graph $H$.  As in Figure~\ref{fig:sft-reduction}, the $\Sigma$ above the self loop on $q_2$ represents a self loop labeled $a$ for each $a \in \Sigma$.}
  \label{fig:sft-h-graph}
\end{figure}

\begin{reduction}\label{red:sft}
  Let $M_1, \dots, M_n$ be an instance to the \dfacupp problem.
  Construct the deterministic presentation $G$ as follows.
  Add a state $t$ (the \term{terminal} state),
  and add a self loop labeled $\star$ on $t$.
  Add a state $s^*$, and add
  self loops on $s^*$ labeled by each symbol in $\Sigma\cup\{\ell\}$.
  Add an edge labeled $\rightm$ from $s^*$ to $t$.
  Finally, for each $i=1, \dots, n$,
  \begin{enumerate}
  \item embed $M_i$ into $G$;
  \item for each state $q \in Q_i$, add an edge labeled $\leftm$ from
    $q$ to $s_i$
  \item for each accepting state $q \in F_i$, add an edge labeled $\rightm$
    from $q$ to $t$;
  \item for each state $q \in Q_i$, add a self loop labeled $\ell$ on $q$.
  \end{enumerate}
  See Figure~\ref{fig:sft-reduction} for a visualization.
\end{reduction}

We once again summarize the salient properties of the reduction. First,
we define a notation that will be used multiple times: for a word
$w$, we let $h_\ell(w)$ denote $w$ with all the $\ell$'s removed.
(That is, $h_\ell$ is the  string homomorphism such that $h_\ell(\ell) = \epsilon$ and $h(a) = a$ for $a \neq \ell$.)

\begin{lemma}\label{lem:red-sft-lemma}
  The following hold of Reduction~\ref{red:sft} and the graph $H$ from Figure~\ref{fig:sft-h-graph}.
  \begin{enumerate}[(i)]
  \item $\shift{G} \subseteq \shift{H}$;
  \item $(\Sigma \cup \{\ell, \leftm\})^*,
    \subseteq F_G(q)$ for all $q \in \bigcup_{i=1}^n Q_i$;
  \item for $w\in\Sigma^*$,
    $w \rightm \in F(s_i)$ if and only if $w \in L(M_i)$;
  \item for $w\in\Sigma^*$,
    $\leftm w \rightm \in \Bc(\shift{G})$ if and only if $w \in \bigcup_{i=1}^n L(M_i)$.
  \item for $w \in (\Sigma \cup \{\ell\})^*$,
    $\leftm h_\ell(w) \rightm \in \Bc(\shift{G})$ if and only if $\leftm w \rightm \in \Bc(\shift{G})$;
  \end{enumerate}
\end{lemma}

\begin{proof}
  For (i), let $w \in \Bc(\shift{G})$. If $w$ does not contain
  $\rightm$, then either $w \in (\Sigma \cup \{\ell, \leftm\})^*$
  or $w = \star^m$ for some $m \geq 0$;
  if $w \in (\Sigma \cup \{\ell, \leftm\})^*$, then
  $w \in F_H(q_1)$;
  if $w = \star^m$ for some $m \geq 0$, then $w \in F_H(q_2)$.
  Otherwise, if $w$ contains $\rightm$, then we can factor
  $w$ into $w = u \rightm \star^m$ where $u \in (\Sigma \cup \{\ell, \leftm\})^*$
  and $m \geq 0$, for which it follows that $w \in F_H(q_1)$. Thus,
  for every $w \in \Bc(\shift{G})$, we have $w \in \Bc(\shift{H})$.  
  
  For the other statements, first note that each of the $M_i$ are emulated
  by the transition action of $G$ in the following way:
  for $q \in Q_i$ and $w \in \Sigma^*$,
  we have $w \in F_G(q)$ and $q \cdot w = \delta_i(q, w)$
  and $q \in F_i$ if and only if $\rightm \in F_G(q)$.
  Thus, (ii) follows from the emulation observation and
  the fact that $\ell \in F_G(q)$ and $q \cdot \ell = q$
  and $\leftm \in F_G(q)$ and $q \cdot \leftm = s_i$
  for all $q \in Q_i$.
  Statement (iii) follows immediately from the emulation observation as well.
  Note that $Q_G \cdot \leftm = \{s_1, \dots, s_n\}$, so
  (iv) follows from (iii). Finally, for (v), as $q \cdot \ell = q$
  for $q \in \bigcup_{i=1}^n Q_i$, by induction
  on the number of $\ell$'s in $w$, one can show that
  $Q_G \cdot \leftm w = Q_G \cdot \leftm h_\ell(w)$;
  thus, we have $Q_G \cdot \leftm h_\ell(w) \rightm = Q_G \cdot \leftm w \rightm$,
  which implies $Q_G \cdot \leftm h_\ell(w) \rightm \neq \varnothing$ if and only if
  $Q_G \cdot \leftm w \rightm \neq \varnothing$.
\end{proof}

To show the correctness of Reduction~\ref{red:irred}, we first give an alternate reduction to \inclusionp.

\begin{theorem}\label{thm:fixed-h-subshift}
  $\bigcup_{i=1}^n L(M_i) = \Sigma^*$ if and only if $\shift{H} \subseteq \shift{G}$.
\end{theorem}

\begin{proof}
  First suppose $\Bc(\shift{H}) \subseteq \Bc(\shift{G})$.
  For every word $w \in \Sigma^*$, we have $\leftm w \rightm \in \Bc(\shift{H})$, giving $\leftm w \rightm \in \Bc(\shift{G})$.
  Lemma~\ref{lem:red-sft-lemma}(iv) now implies $\bigcup_{i=1}^n L(M_i) = \Sigma^*$.
  For the converse, suppose $\bigcup_{i=1}^n L(M_i) = \Sigma^*$, and let $u \in \Bc(\shift{H})$.
  There are two cases: either $u \in F_H(q_1)$ or $u \in F_H(q_2)$. If
  $u \in F_H(q_2)$, then $u = \star^m$ for some $m \geq 0$, which implies
  that $u \in F_G(t)$, and so $u \in \Bc(\shift{G})$. Thus, to complete
  the proof we need to show that if $u \in F_H(q_1)$, then $u \in \Bc(\shift{G})$.

  Suppose $u \in F_H(q_1)$. We further break this case into the possible values of $q_1 \cdot u$.
  If $q_1 \cdot u = q_1$, then $u \in (\Sigma \cup \{\ell, \leftm\})^*$, so $u \in F_G(s_1)$ and thus $u \in \Bc(\shift{G})$.
  If $q_1 \cdot u = q_2$, then $u = v \rightm \star^m$ for some $m \geq 0$ and $v \in (\Sigma \cup \{\ell, \leftm\})^*$.
  If $v$ does not contain the symbol $\leftm$, then $v \in (\Sigma \cup \{\ell\})^*$, which implies $v \rightm \star^m = u \in F_G(s^*)$ and thus $u \in \Bc(\shift{G})$.
  Otherwise, if $v$ contains the symbol $\leftm$, then we can factor $v$ into $v=x \leftm w$ where $w$ contains no $\leftm$; i.e. $w \in (\Sigma \cup \{\ell\})^*$.
  Then, we have that $h_\ell(w) \in \Sigma^*$,
  and as $\bigcup_{i=1}^n L(M_i) = \Sigma^*$,
  we have $\leftm h_\ell(w) \rightm \in \Bc(\shift{G})$.
  By Lemma~\ref{lem:red-sft-lemma}(v), we have $\leftm w \rightm \in \Bc(\shift{G})$,
  and by Lemma~\ref{lem:red-sft-lemma}(iii), we have $w \rightm \in F_G(s_i)$ for some $s_i$.
  Collecting facts, we have $x \leftm \in F_G(s_i)$ and $s_i \cdot x \leftm = s_i$
  and $s_i \cdot w \rightm = t$ and $*^m \in F_G(t)$.
  Combining those facts gives us that $x \leftm w \rightm \star^m = u \in F_G(s_i)$, so
  $u \in \Bc(\shift{G})$.
\end{proof}

\begin{remark}
  \label{remark:subshift-sdp-hard}
  Interestingly, Theorem~\ref{thm:fixed-h-subshift} gives us another proof that
  \inclusionp is \cc{PSPACE}-hard.
  However, we can
  easily extend Theorem~\ref{thm:fixed-h-subshift} to a stronger hardness result.\footnote{
    In fact, there is another hardness result proved by the previous
    theorem: \inclusionp is \cc{PSPACE}-hard even when the first argument is
    fixed. That is, for each deterministic presentation $H$, it is \cc{PSPACE}-hard
    to decide when given a deterministic presentation $G$ whether
    $\shift{H} \subseteq \shift{G}$.}
  Specifically,
  we can show that \inclusionp is \cc{PSPACE}-hard even when
  both input instances are synchronizing deterministic presentations.
  Note that $H$ is synchronizing while $G$ is not. Construct the
  presentation $G'$ as follows: construct $G$, and let
  $S \triangleq \{s_1, \dots, s_n, s^*\}$. For each
  $q \in S$, add a self loop labeled $\ell_q$ on $q$.

  For each vertex in $q \in S$, we have that $\ell_q$
  synchronizes to $q$ in $G'$. Note that every vertex
  is reachable from a vertex in $S$, so this implies
  that $G'$ is synchronizing. Here, we note that $\shift{H} \subseteq \shift{G}$ if and only if
  $\shift{H} \subseteq \shift{G'}$: the forward
  direction follows from the fact that $\shift{G} \subseteq \shift{G'}$,
  and the reverse direction follows from the fact that if $w \in \Bc(\shift{H})$
  and $w \in \Bc(\shift{G'})$,
  then $w$ does not contain the new labels $\{\ell_{s_1}, \dots, \ell_{s_n}, \ell_{s^*}\}$
  added in $G'$, so it must be the case that $w \in \Bc(\shift{G})$.
  This establishes the claim that \inclusionp is \cc{PSPACE}-hard
  even when both instances are synchronizing deterministic presentations.
\end{remark}

As $\shift{G} \subseteq \shift{H}$ by Lemma~\ref{lem:red-sft-lemma}(i), we have the following.

\begin{corollary}\label{cor:sft-eq}\label{cor:equal-hard}
  $\bigcup_{i=1}^n L(M_i) = \Sigma^*$ if and only if $\shift{H} = \shift{G}$.
\end{corollary}

We new show that Reduction~\ref{red:sft} reduces \dfacupp to \sftp.

\begin{theorem}\label{thm:sft-hard}
    $\bigcup_{i=1}^n L(M_i) = \Sigma^*$ if and only if $\shift{G}$ is an SFT.
\end{theorem}

\begin{proof}
  The edges in $H$ are labeled uniquely, so
  by Lemma~\ref{lem:sft-edge}, $\shift{H}$ is an SFT.
  Thus, if we have $\bigcup_{i=1}^n L(M_i) = \Sigma^*$,
  then by Corollary~\ref{cor:sft-eq}, $\shift{G} = \shift{H}$ is an SFT.

  Conversely, suppose $\shift{G}$ is an SFT. By Theorem~\ref{thm:sft-char}, there is an $M$
  such that whenever
  $uv, vw \in \Bc(\shift{G})$ and $|v| \geq M$, then $uvw \in \Bc(\shift{G})$.
  Let $w \in \Sigma^*$. We can find $\leftm w \ell^M \in F_G(s_1)$
  and $w \ell^M \rightm \in F_G(s^*)$, so we have
  $\leftm w \ell^M,\, w \ell^M \rightm \in \Bc(\shift{G})$ and
  and thus $\leftm w \ell^M \rightm \in \Bc(\shift{G})$.
  As $h_\ell(w\ell^M) = w$, Lemma~\ref{lem:red-sft-lemma}(v) implies that
  $\leftm w \rightm \in \Bc(\shift{G})$.
  It follows that $w \in \bigcup_{i=1}^n L(M_i)$ by Lemma~\ref{lem:red-sft-lemma}(iv).
\end{proof}

Along with Corollary~\ref{cor:sft-eq}, the following shows that Reduction~\ref{red:sft} also reduces \dfacupp to \minimalp.

\begin{theorem}\label{thm:minimal-hard}
  $\bigcup_{i=1}^n L(M_i) = \Sigma^*$ if and only if
  $\shift{G}$ has a deterministic presentation with $2$ vertices.
\end{theorem}

\begin{proof}
  If $\bigcup_{i=1}^n L(M_i) = \Sigma^*$,
  then $\shift{H} = \shift{G}$ by Corollary~\ref{cor:sft-eq},
  so $H$ is a $2$-vertex presentration of $\shift{G}$.
  Conversely, suppose $H'$ is a deterministic presentation of $\shift{G}$ with $2$ vertices.
  Here, we'll show that $H'$ is isomorphic to $H$, which implies
  $\shift{H'} = \shift{H}$ and thus $\shift{H} = \shift{G}$, so
  by Corollary~\ref{cor:sft-eq}, we have $\bigcup_{i=1}^n L(M_i) = \Sigma^*$.

  As $\rightm \in F_G(s^*)$, there must be an edge $e_\rightm$ labeled $\rightm$ in $H'$.
  If $e_\rightm$ were a self loop, then $\rightm \rightm \in \Bc(\shift{G})$, a contradiction.
  Thus, the vertices $q_1' \triangleq i(e_\rightm)$ and $q_2' \triangleq t(e_\rightm)$ must be distinct,
  and $Q_{H'} = \{q_1', q_2'\}$.
  Moreover, $e_\rightm$ must be the unique edge labeled $\rightm$, as in all other cases $H'$ either fails to be deterministic or we again have $\rightm \rightm \in \Bc(\shift{G})$, a contradiction.
  
  As $a \rightm \in \Bc(\shift{G})$,
  for each $a \in \Sigma \cup \{\ell\}$, we have
  an edge $e_a$ labeled $a$ ending at $q_1'$.
  For $a \in \Sigma \cup \{\ell\}$,
  any edge labeled $a$ must start at $q_1'$:
  if such an edge started at $q_2'$, then $\rightm a \in \Bc(\shift{G})$, a contradiction.
  Thus, $e_a$ is a self loop and by determinism,
  $e_a$ is the unique edge labeled $a$ in $H'$.

  As $\rightm \star \in \Bc(\shift{G})$, there must be
  an edge $e_\star$ labeled $\star$ starting at $q_2'$.
  If $e_\star$ ends at $q_1'$, then $\star \rightm \in \Bc(\shift{G})$,
  a contradiction, so $e_\star$ is a self loop.
  Any edge labeled $\star$ must start at $q_2'$:
  if such an edge started at $q_1'$, then $\ell \star \in \Bc(\shift{G})$, a contradiction.
  Thus, by determinism, $e_\star$ is the unique edge labeled $\star$ in $H'$.

  Finally, as $\leftm \ell \in \Bc(\shift{G})$, there must be some
  edge $e_\leftm$ labeled $\leftm$ ending at $q_1'$. Any edge labeled
  $\leftm$ must start at $q_1'$: if such an edge start at $q_2'$, then
  $\rightm \leftm \in \Bc(\shift{G})$, a contradiction. Thus,
  $e_\leftm$ is a self loop and is the unique edge labeled $\leftm$ in
  $H'$.  All of the above implies that the map
  $q_i' \mapsto q_i$ is an isomorphism between $H'$ and $H$.
\end{proof}

\subsection{Hardness of Existence of Synchronizing Words}
\label{sec:hard-sync-word}

\citet{berlinkov2014two} showed \syncp was \cc{PSPACE}-hard via reduction
from the \cc{PSPACE}-complete problem of \term{subset
synchronizability}: given a DFA $M$ and a subset $S\subseteq Q_M$, is
there a word $w$ such that $|S \cdot w|=1$?
For completeness, we show the hardness of $\syncp$ via Reduction~\ref{red:sync} from the ``complement'' of \dfacupp, \dfacapp.

\begin{theorem}\label{thm:exist-sync-hard}
  \syncp is \cc{PSPACE}-hard.
\end{theorem}

For the reduction, we again create pre-initial states $p_i$ for each DFA $M_i$, with special symbols $\leftm$ into and $\rightm$ out of each DFA, and include them in parallel as in Reduction~\ref{red:sft}.
The edges out of accepting states all go to the same shared succes state $t$.
We also add edges labeled $\rightm$ from each nonaccepting state in $M_i$ to an individual fail state $r_i$.
By completing this construction appropriately, we ensure that a word is synchronizing if and only if it is synchronizing to $t$, i.e., if and only if every DFA accepts the subword between $\leftm$ and $\rightm$.

\begin{reduction}\label{red:sync}
  Let $M_1, \dots, M_n$ be an instance to the \dfacapp problem, and
  without loss of generality, assume $n \geq 2$.
  We 
  will construct a essential, deterministic presentation $G$ as follows. 
  First, add a state $t$ (the \term{success} state), and add 
  a self loop labeled $\rightm$ on $t$. Then, for each $i=1, \dots, n$,
  \begin{enumerate}
  \item add a state $p_i$ (the $i$th \term{pre-initial} state);
  \item add a state $r_i$ (the $i$th \term{fail} state);
  \item add self loops labeled $\rightm$ on $p_i$ and $r_i$;
  \item for each $a \in \Sigma$, add a self loop labeled $a$ on $p_i$;
  \item embed $M_i$ into $G$;
  \item add an edge from $p_i$ labeled $\leftm$ to the corresponding initial state $s_i$
    of $M_i$;
  \item for each accepting state $q$ in $M_i$, add an edge labeled $\rightm$
    from $q$ to $t$;
  \item for each nonaccepting state $q$ in $M_i$, add an edge labeled $\rightm$
    from $q$ to $r_i$.
  \end{enumerate}
  See Figure~\ref{fig:sync-red} for a visualization.
\end{reduction}

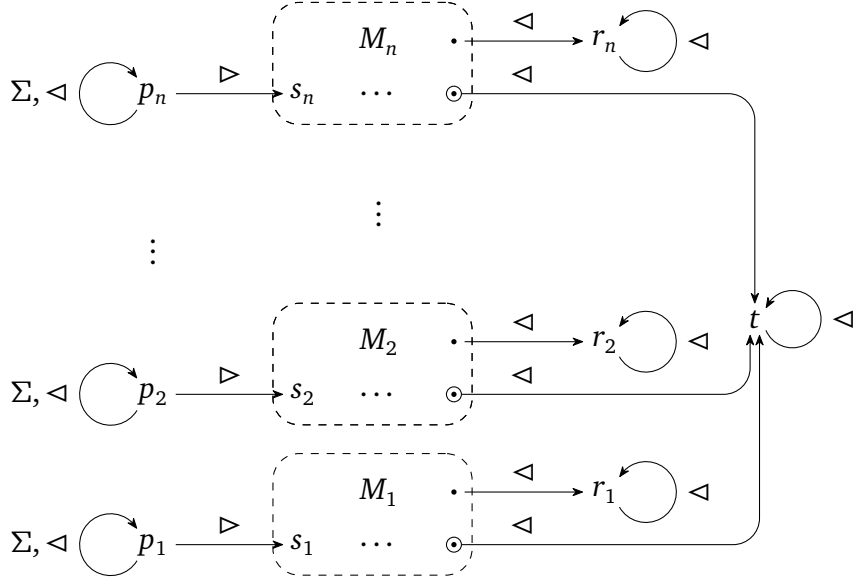
\begin{figure}[t]
  \centering
  {
    \begin{tikzpicture}[rounded corners=10pt]
      \path (0, 0) node[vertex] (P1) {$p_1$};
      
      \fill (P1) ++(\xgap, 0) node[vertex] (s1) {$s_1$}
      +(0.5*\xgap, 0.35*\ygap) node[vertex] (M1) {$M_1$}
      +(0.5*\xgap, 0) node {$\dots$}
      +(\xgap, 0) circle[radius=1pt] node[draw, circle, inner sep=2pt] (f1) {}
      +(\xgap, .35*\ygap) circle[radius=1pt] node (nf1) {};
      \node[draw, dashed, rectangle, fit=(M1) (s1) (f1)] {};

      \path (nf1) ++(\xgap, 0) node[vertex] (F1) {$r_1$};

      \path (0, \ygap) node[vertex] (P2) {$p_2$};

      \fill (P2) ++(\xgap, 0) node[vertex] (s2) {$s_2$}
      +(0.5*\xgap, 0.35*\ygap) node[vertex] (M2) {$M_2$}
      +(0.5*\xgap, 0) node {$\dots$}
      +(\xgap, 0) circle[radius=1pt] node[draw, circle, inner sep=2pt] (f2) {}
      +(\xgap, .35*\ygap) circle[radius=1pt] node (nf2) {};
      \node[draw, dashed, rectangle, fit=(M2) (s2) (f2)] {};

      \path (nf2) ++(\xgap, 0) node[vertex] (F2) {$r_2$};

      \path (1.5*\xgap , 2.25*\ygap) node[vertex] (dots) {$\vdots$};
      \path (0, 2*\ygap) node[vertex] (Pdots) {$\vdots$};

      \path (0, 3*\ygap) node[vertex] (Pn) {$p_n$};
      
      \fill (Pn) ++(\xgap, 0) node[vertex] (sn) {$s_n$}
      +(0.5*\xgap, 0.35*\ygap) node[vertex] (Mn) {$M_n$}
      +(0.5*\xgap, 0) node {$\dots$}
      +(\xgap, 0) circle[radius=1pt] node[draw, circle, inner sep=2pt] (fn) {}
      +(\xgap, .35*\ygap) circle[radius=1pt] node (nfn) {};
      \node[draw, dashed, rectangle, fit=(Mn) (sn) (fn)] {};

      \path (nfn) ++(\xgap, 0) node[vertex] (Fn) {$r_n$};

      \path (fn) ++(2*\xgap, -1.5*\ygap) node[vertex] (S) {$t$};

      \draw (f1) -| ($(S.south) + (0.5*\arrowgap, 0)$);
      \draw (f2) -| ($(S.south) - (0.5*\arrowgap, 0)$);
      \draw (fn) -| (S.north);
      
      \draw (nf1) to node[above] {$\rightm$} (F1);
      \draw (nf2) to node[above] {$\rightm$} (F2);
      \draw (nfn) to node[above] {$\rightm$} (Fn);
      
      \path (f1) +(\xgap, 0) node[vertex, invisible] (iF1) {$r_1$}
      (f1) to node[above] {$\rightm$} (iF1);
      \path (f2) +(\xgap, 0) node[vertex, invisible] (iF2) {$r_2$}
      (f2) to node[above] {$\rightm$} (iF2);
      \path (fn) +(\xgap, 0) node[vertex, invisible] (iFn) {$r_n$}
      (fn) to node[above] {$\rightm$} (iFn);

      \draw (F1.south east) to[right loop] node [right] {$\rightm$} (F1.north east);
      \draw (F2.south east) to[right loop] node [right] {$\rightm$} (F2.north east);
      \draw (Fn.south east) to[right loop] node [right] {$\rightm$} (Fn.north east);
      \draw (S.south east) to[right loop] node [right] {$\rightm$} (S.north east);

      \draw (P1) to node[above] {$\leftm$} (s1);
      \draw (P2) to node[above] {$\leftm$} (s2);
      \draw (Pn) to node[above] {$\leftm$} (sn);

      \node[draw, dashed, rectangle, fit=(M2) (s2) (f2)] {};
      \node[draw, dashed, rectangle, fit=(Mn) (sn) (fn)] {};

      \draw (P1.south west) to[left loop] node[left] {$\Sigma, \rightm$} (P1.north west);
      \draw (P2.south west) to[left loop] node[left] {$\Sigma, \rightm$} (P2.north west);
      \draw (Pn.south west) to[left loop] node[left] {$\Sigma, \rightm$} (Pn.north west);
    \end{tikzpicture}
  }
  \caption{Schematic of Reduction~\ref{red:sync}.}
  \label{fig:sync-red}
\end{figure}

To show the correctness of the reduction, we characterize the synchronizing words of $G$.

\begin{theorem}\label{thm:sync-word-characterization}
  Let $G$ be the deterministic presentation obtained from Reduction
  \ref{red:sync} on an instance $M_1, \dots, M_n$.
  A word $u\in(\Sigma \cup \{\leftm,\rightm\})^*$ is synchronizing for $G$ if and only if there is some $v \in (\Sigma \cup \{\rightm\})^*$, $k \geq 1$, and $w \in \bigcap_{i=1}^n L(M_i)$ such that $u=v \leftm w \rightm^k$.
\end{theorem}

\begin{proof}
  As usual, the transition action of $G$ emulates the behavior of the $M_i$:
  for $w \in \Sigma^*$, we have
  $w \in L(M_i)$ if and only if $p_i \cdot \leftm w \rightm = t$,
  and $w \notin L(M_i)$ if and only if $p_i \cdot \leftm w \rightm = r_i$.

  Suppose $u$ is a synchronizing word for $G$.
  Then, $u$ must contain at least one $\leftm$;
  otherwise $u \in (\Sigma \cup \{\rightm\})^*$,
  and thus $p_i \cdot u = p_i$ for each $i$, giving $|Q_G \cdot u | \geq n \geq 2$.
  We can therefore write $u=u_1 \leftm u_2$.
  By construction, $u$ contains at most one $\leftm$, so $u_1,u_2 \in (\Sigma \cup \{\rightm\})^*$.
  Moreover, we must have $u_2 = w \rightm^k$ for some $w \in \Sigma^*$ and $k \geq 0$.
  Since $Q_G \cdot u_1 \leftm w = \{p_1 \cdot \leftm w, \dots, p_n \cdot \leftm w\}$
  and the $Q_i$ are pairwise disjoint, we have $|Q_G \cdot u_1 \leftm w| = n \geq 2$.
  Since $u$ is synchronizing, we therefore must have $k \geq 1$. 
  Now since $\rightm \in F_G(q)$ for $q \in \bigcup_{i=1}^n Q_i$,
  if there are $i \neq j$ such that $p_i \cdot \leftm w \rightm$
  and $p_j \cdot \leftm w \rightm$ are not both $t$, then $|Q_G \cdot u| \geq 2$.
  As $u$ is synchronizing, we conclude $p_i \cdot \leftm w \rightm = t$ for all $i$.
  By the above, $w \in \bigcap_{i=1}^n L(M_i)$.
  Hence, we have $u = u_1 \leftm w \rightm^k$ with
  $u_1 \in (\Sigma \cup \{\rightm\})^*$, $k \geq 1$,
  and $w \in \bigcap_{i=1}^n L(M_i)$.

  Conversely, let $v \in (\Sigma \cup \{\rightm\})^*$, $k \geq 1$,
  and $w \in \bigcap_{i=1}^n L(M_i)$.
  Thus,
  $p_i \cdot \leftm w \rightm = t$ for all $i$, which implies
  $Q_G \cdot v \leftm w \rightm^k = \{t\}$. 
\end{proof}

\begin{corollary}
   $\bigcap_{i=1}^n L(M_i) \neq \varnothing$ if and only if $G$ has a synchronizing word.
\end{corollary}

As Reduction~\ref{red:sync} therefore reduces \dfacapp to \syncp, Theorem~\ref{thm:exist-sync-hard} follows.

\section{Size of Synchronizing Words and SDPs}
\label{sec:size-synchr-words-sdps}

Our reductions also shed light on the size of synchronizing words and presentations.
In particular, given a presentation with $n$ vertices, the size of its smallest synchronizing word can be exponentially large in $n$.
Similarly, the size of the smallest synchronizing deterministic presentation can also be exponentially large.

\subsection{Shortest synchronizing word size}
\label{sec:short-synchr-word}

If $\cc{NP} \neq \cc{PSPACE}$, there cannot be a polynomial
upper bound with respect to the number of vertices for the length
of the shortest synchronizing word, as \syncp is \cc{PSPACE}-hard.
\citet{berlinkov2014two} show an unconditional exponential lower bound on maximum length of the shortest synchronizing word, which implies there cannot be a
polynomial upper bound for the length of the shortest synchronizing word.
Here, we give a simpler construction that achieves roughly the same bound.

First we observe the following property of Reduction~\ref{red:sft}.

\begin{lemma}\label{lem:sw-length}
  Let $G$ be a presentation obtained from Reduction~\ref{red:sft} on some input $M_1, \dots, M_n$.
  If $\bigcap_{i=1}^n L(M_i) \neq \varnothing$, then the minimum length of a
  synchronizing word for $G$ is 2 more than the minimum
  length of a word in $\bigcap_{i=1}^n L(M_i)$.
\end{lemma}

\begin{proof}
  From Theorem~\ref{thm:sync-word-characterization}, a word $u$ is synchronizing for $G$ if and only if it has the form $u = v \leftm w \rightm^k$ for any $v \in (\Sigma \cup \{\rightm\})^*$, $w \in \bigcap_{i=1}^n L(M_i)$, and $k \geq 1$.
  A minimum-length synchronizing word $u^*$ for $G$ therefore has $v=\epsilon$ and $k=1$, and takes $w=w^*$ to be a word of minimum length in $\bigcap_{i=1}^n L(M_i)$.
  Thus $|u^*| = |\leftm w^* \rightm| = 2 + |w^*|$.
\end{proof}

Therefore, to find a presentation of a sofic shift with a large shortest synchronizing word, it suffices to apply Reduction~\ref{red:sft} to DFAs that have a large shortest word in the intersection of their languages.
In Appendix~\ref{apx:cap-cons}, we adapt a construction from
\citet{ang2010problems} of a family of DFAs $M_{i,k}$ such that each
$M_{i,k}$ has 3 states and and the shortest word in
$\bigcap_{i=0}^k L(M_{i,k})$ is $2^k$.
Using this family of DFAs, if
we let $G_k$ denote Reduction~\ref{red:sft} applied to
$M_{0,k}, M_{1,k}, \dots M_{k, k}$, then by Lemma~\ref{lem:sw-length},
the shortest synchronizing word for $G_k$ has length $2^k+2$.
The number of vertices in $G_k$ is $2(k+1)+1$ auxillary vertices plus $3(k+1)$
from the DFAs, giving use $5k+6$ total vertices.
We may then define a family of graphs $G^{(n)}$ on $n$ vertices, which take $G_k$ where $k \triangleq \left\lfloor \frac{n-6}{5} \right\rfloor$ and add $n-k$ vertices without affecting the shortest synchronizing word (e.g., by adding a self loops labeled with every
$a \in \Sigma \cup \{\leftm\}$ and adding an edge
labeled $\rightm$ to $t$).
As $k = \Omega(n)$, this family exhibits the following lower bound.

\begin{theorem}
  \label{thm:min-sync-word-lower-bound}
  There is a family of deterministic presentaions $\{G^{(n)}\}$ such that
  for sufficiently large $n$, each $G^{(n)}$ has $n$ vertices and the minimum length
  of a synchronizing word for $G^{(n)}$ is $2^{\Omega(n)}$.
\end{theorem}

\begin{remark}
  \Cerny's conjecture states that if a $n$-state DFA has a
  synchronizing word, then there is one of length at most $(n-1)^2$
  \cite{volkov2008synchronizing}.
  The previous theorem is a counterexample to the generalization of
  the \Cerny's conjecture to deterministic presentations of sofic shifts.
\end{remark}

\subsection{Minimal Synchronizing Deterministic Presentation Size}
\label{sec:minim-synchr-determ}

Throughout this subsection, we will abbreviate ``synchronizing deterministic presentation'' to SDP.
Let $X$ be a sofic shift.
A \term{minimal SDP} of $X$ is a SDP of $X$ possessing the fewest number of vertices among all SDPs of $X$.
Similarly, a \term{minimal deterministic presentation} of $X$ is a deterministic presentation of $X$ possessing the fewest number of vertices among all deterministic presentations of $X$.
For a given sofic shift $X$, minimal SDPs of $X$ are unique up to isomorphism, while minimal deterministic presentations are not neccessarily unique if $X$ is reducible \cite{jonoska1996sofic}.
For irreducible sofic shifts, minimal SDPs are minimal derterministic presentations and vice versa.
For reducible sofic shifts, minimal SDPs are minimal deterministic presentations are not necessarily the same.

In fact, we show that the minimal SDP can be exponential larger than a minimal deterministic presentation.
The proof relies on multiple-entry DFAs, which are FAs whose only nondeterminism is the fact that there are multiple possible initial states.
Formally, a \term{$k$-entry DFA} $N$ is a fully deterministic labeled
graph along with $k$ states $s_1, \dots, s_k \in Q_N$ and a set of
final states $F \subseteq Q_N$.
The language of $N$ is defined as
$L(N) \triangleq \bigcup_{i=1}^k \{w \in \Ac_N^* : \delta(s_i, w) \in
F\}$.
By \citet[Lemma 3]{holzer2001state}, there exists a family $\{C_k\}$ of multiple-entry DFAs, where $C_k$ is a $k$-entry DFA with $k$ states, and and the minimal DFA for $L(C_k)$ has $\sum_{i=1}^k {k \choose i} = 2^k-1$ states.
In other words, even if automata have deterministic transition relations, passing from multiple start states to a single start state can incur an exponential increase in size.
We emulate this construction when passing from a nonsynchronizing presentation to a synchronizing presentation.

Given a $k$-entry DFA $N$, we construct a sofic shift
$\shift{G}$ with deterministic presentation $G$ as follows.
First,
embed $N$ into $G$, add a state $t$ (the \term{terminal} state),
and add a self loop labeled $\star$ on $t$.  Then, for each
$i=1, \dots, k$,
\begin{enumerate}
    \item add a state $p_i$ (the $i$th \term{pre-initial} state);
    \item add a self loop labeled $*$ on $p_i$;
    \item add an edge from $p_i$ labeled $\leftm$ to the
      corresponding state $s_i$ of $N$;
    \item for each accepting state $q \in F$, add an edge labeled $\rightm$
      from $q$ to $t$
\end{enumerate}
See Figure~\ref{fig:sdp-construction} for a visualization.

\begin{figure}[t]
  \centering
  {
    \begin{tikzpicture}[rounded corners=10pt]

      \path (\xgap , 2*\ygap) node[vertex] {$\vdots$};
      \path (0, 2*\ygap) node[vertex] {$\vdots$};

      \path (0, 0) node[vertex] (P1) {$p_1$}
      ++(\xgap , 0) node[vertex] (s1) {$s_1$}
      ++(.5*\xgap, 0) node[vertex] {$\dots$};
      
      \path (0, \ygap) node[vertex] (P2) {$p_2$} 
      ++(\xgap , 0) node[vertex] (s2) {$s_2$}
      ++(.5*\xgap, 0) node[vertex] {$\dots$};
    
      \path (0, 3*\ygap) node[vertex] (Pn) {$p_k$} 
      ++(\xgap , 0) node[vertex] (sn) {$s_k$}
      ++(.5*\xgap, 0) node[vertex] {$\dots$}
      ++(0, 0.35*\ygap) node[vertex] (M) {$N$};

      \fill (2*\xgap, 1.5*\ygap) circle[radius=1pt] node[draw, circle, inner sep=2pt]
      (f) {};

      \path (3*\xgap, 1.5*\ygap) node[vertex] (T) {$t$};
      \draw (T.south east) to[right loop] node[right] {$\star$} (T.north east);

      \draw (f) to node[above] {$\rightm$} (T);
      
      \draw (P1) to node[above] {$\leftm$} (s1);
      \draw (P2) to node[above] {$\leftm$} (s2);
      \draw (Pn) to node[above] {$\leftm$} (sn);

      \node[draw, dashed, rectangle, inner xsep=10pt, fit=(M) (f) (s1)] {};

      \draw (P1.south west) to[left loop] node[left] {$*$} (P1.north west);
      \draw (P2.south west) to[left loop] node[left] {$*$} (P2.north west);
      \draw (Pn.south west) to[left loop] node[left] {$*$} (Pn.north west);
    \end{tikzpicture}
  }
  \caption{Schematic of the construction from Theorem~\ref{thm:k-entry-construction}.}
  \label{fig:sdp-construction}
\end{figure}
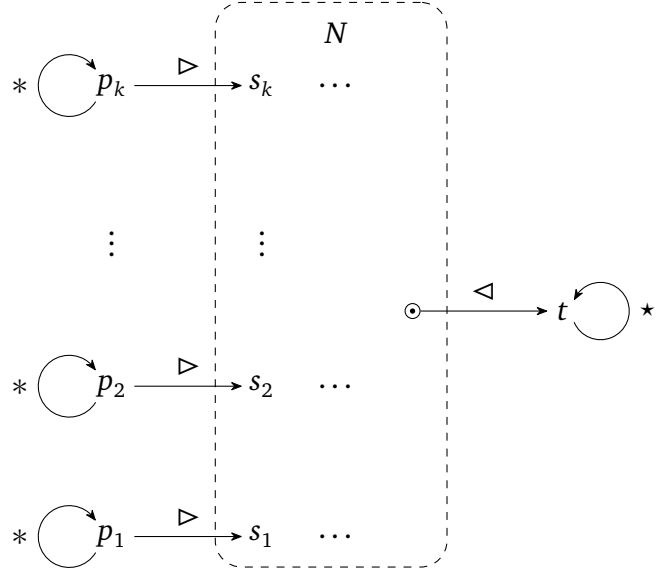

\begin{theorem}\label{thm:k-entry-construction}
  Let $N$ be a $k$-entry DFA, and let $G$ be the deterministic
  presentation obtained from the above construction applied to $N$.
  Let $M$ be the minimal DFA of $L(N)$.
  Interpreting $M$ as a $1$-entry DFA, let $H$ be the deterministic presentation
  obtained from the construction applied to $M$.
  Then, $H$ is the minimal SDP for $\shift{G}$.
\end{theorem}

\begin{proof}
  The proof that $\shift{G} = \shift{H}$ is similar to the proof of Theorem~\ref{thm:fixed-h-subshift}.
  To show $H$ is the minimal SDP
  for $\shift{G}$, by Jonoska \cite[Theorem 5.5]{jonoska1996sofic},
  it suffices to show $H$ is follower-separated.
  For distinct $p, q \in Q_M$, as $M$ is a minimal DFA,
  either
  there is some word $\delta(p, w) \in F$ with
  $\delta(q, w) \notin F$ or
  there is some word $\delta(p, w) \notin F$ with
  $\delta(q, w) \in F$.
  Without loss of generality, we may assume the former, as
  we can swap the roles of $p$ and $q$ for the latter case.
  Then, $w \rightm \in F_H(p)$ while
  $w \rightm \notin F_H(q)$, so $p$ and $q$ have distinct follower sets.
  For distinct $p, q \in Q_H$ where one of $p$ or $q$ is the pre-initial state or
  the terminal state, follower-separation follows from the presence of
  $*$ or $\star$.
\end{proof}

Thus, the size of the minimal SDP
for $\shift{G}$ is determined by the size of the minimal DFA for
$L(N)$. Applying this construction to $C_k$ gives us a deterministic
presentation with $2k+1$ vertices whose minimal synchronizing deterministic
presentation with $(2^k-1)+2 = 2^k+1$ vertices. The following theorem follows
easily from this observation. 

\begin{theorem}\label{thm:sdp-lower-bound}
  There is a family of sofic shifts $\{X_n\}$ such that for
  sufficiently large $n$,  the minimal deterministic presentation of
  $X_n$ has at most $n$ vertices and the minimal synchronizing deterministic
  presentation of $X_n$ has $2^{\Omega(n)}$ vertices.
\end{theorem}

\begin{proof}
  Let $n$ be sufficiently large. If we apply the construction to
  $C_{n'}$ where
  $n' \triangleq \left\lfloor \frac{n-1}{2} \right\rfloor$, we get a
  presentation $G$ with at most $n$ vertices such that the minimal
  SDP for $\shift{G}$ has
  $2^{n'}+1$ vertices. Thus, since $G$ has at most $n$ vertices, then
  a minimal deterministic presentation for $\shift{G}$ must have at
  most $n$ vertices, and as $n' = \Omega(n)$, the minimal
  SDP for $\shift{G}$ has
  $2^{\Omega(n)}$ vertices.
\end{proof}

\section{Discussion}

We first overview our results, together with a discussion of related problems and a comparison to results from automata theory.
We conclude with open problems.

\subsection{Overview of results}

We summarize our results in Table~\ref{tab:results-overview}.
The table includes complexity results from the automata theory literature for analagous problems for DFAs and FAs.
One conclusion from this table is that, from a computational complexity standpoint, irreducible presentations behave like DFAs, as do synchronizing deterministic presentations with the exception of \inclusionp.
On the other hand, non-deterministic presentations behave like NFAs, as do general deterministic presentations, with the exception of \problem{Universality}.

\newcommand{\tabP}{\cc{P}}
\newcommand{\tabH}{\cc{PSP}-c}

\begin{table}[t]
  \centering
  \begin{tabular}{l@{\hskip 20pt}ll@{\hskip 40pt}llll}
    \toprule
    & \multicolumn{2}{c}{Automata~~~~~~~~~~~~~}
    & \multicolumn{4}{c}{Sofic Shifts~~~~~~}
    \\
    Problem       & DFA   &  FA   & IDP   & SDP   & DP & P   \\
    \midrule                      
    \problem{Universality}  & \tabP & \tabH & \tabP & \tabP & \tabP & \tabH~\cite{czeizler2006testing} \\
    \equalp       & \tabP & \tabH & \hyperref[thm:subshift-correct]{\tabP} & \hyperref[sec:iso-eq]\tabP & \hyperref[thm:hardness-four]{\tabH} & \tabH \\
    \inclusionp   & \tabP & \tabH & \hyperref[thm:subshift-correct]{\tabP} & \hyperref[remark:subshift-sdp-hard]{\tabH} & \hyperref[thm:hardness-four]{\tabH} & \tabH \\
    \minimalp     & \tabP & \tabH & \tabP~\cite{lind2021introduction} & \tabP~\cite{jonoska1996sofic} & \hyperref[thm:sft-minimal-hard]{\tabH} & \tabH \\
    \syncp        & \tabP~\cite{eppstein1990reset}      &    & \hyperref[thm:sw-correct]{\tabP} &       & \hyperref[thm:exist-sync-hard]{\tabH}~\cite{berlinkov2014two} & \tabH \\
    \midrule                      
    \irredp       &       &       &       & \tabP & \hyperref[thm:hardness-four]{\tabH} & \tabH \\
    \sftp         &       &       & \hyperref[thm:sft-testing-sdp]{\tabP}~\cite{lind2021introduction} & \hyperref[thm:sft-testing-sdp]{\tabP} & \hyperref[thm:sft-minimal-hard]{\tabH} & \tabH \\
    \sdpp         &       &       &       &       & \hyperref[thm:hardness-four]{\tabH} & \tabH\\
    \bottomrule
  \end{tabular}
  \caption{An overview of our results and comparison to related automata theory results. The classes sofic shifts are as follows:
    IDP = irreducible deterministic presentation,
    SDP = synchronizing deterministic presentation,
    DP = (general) deterministic presentation,
    and
    P = (general) presentation.
    The complexity classes are
    \tabP\ = solvable in polynomial time
    and \tabH\ = \cc{PSPACE}-complete.
    For FAs, \inclusionp means ``is $L(M) \subseteq L(N)$?''
    Entries of the table corresponding to results we prove (or re-prove) have hyperlinks to their respective proofs.
  }\label{tab:results-overview}
\end{table}

The remainder of this subsection is devoted to entries of the table which were not discussed in the previous sections.
To begin, the problem of \problem{Universality} asks whether a given deterministic presentation $G$ satisfies $\shift{G} = \Ac_G^\Zb$.
Clearly, we have $\shift{G}
\subseteq \Ac_G^\Zb$ for any $G$, so the problem reduces to deciding whether
$\Ac_G^\Zb \subseteq \shift{G}$. This condition can be decided in polynomial time,
as there is an irreducible deterministic presentation of $\Ac_G^\Zb$
which is a single vertex and a self loop on that vertex labeled $a$
for each $a \in \Ac_G$, so one may use Algorithm \ref{alg:subshift} to
decide whether $\Ac_G^\Zb \subseteq \shift{G}$. However, for nondeterministic
presentations, universality is \cc{PSPACE}-complete
\cite{czeizler2006testing}.
\problem{Universality} is equivalent to \minimalp for $k=1$, i.e.\ deciding if the sofic shift presented by a determinstic presentation has a $1$-vertex presentation,
as $\shift{G} = \Ac_G^\Zb$ exactly when $\shift{G}$ has a $1$-vertex presentation. Thus, \minimalp for $k=1$ is in \cc{P}, and our results show that \minimalp for $k\geq 2$ is \cc{PSPACE}-complete.
(Our reduction shows hardness for $k=2$; simple modifications give $k > 2$.)

Given a synchronizing determinstic
presentation, \syncp is trivial, as a synchronizing word always exists.
To actually find a synchronizing word in this case, however, Algorithm \ref{alg:sw} is not sufficient:
for reducible presentations, the algorithm can fail when there are two vertices with the same follower set but no word sending them to the same vertex (e.g.\ $r_1$ and $r_2$ in Reduction \ref{red:sync}).
Fortunately,
the proof of Theorem \ref{thm:sync-testing-alg} gives a method of
constructing a word that synchronizes to any vertex: find a synchronizing word for an initial irreducible component, a word that separates it from the rest of the graph, and finally a word leading to the desired vertex.
This procedure can be
implemented using Algorithm \ref{alg:sw} and Algorithm
\ref{alg:subshift} using only polynomial time.

Similarly, \irredp is trivial for irreducible deterministic presentations, as the shift is guaranteed to be irreducible.
For a synchronizing deterministic presentation $G$, \irredp can be decided by testing whether $G$ is irreducible. In particular, if $\shift{G}$ is irreducible, then by
Theorem \ref{thm:irred-char}, the subgraph induced by all the
synchronizing vertices is irreducible; as every vertex is
synchronizing, $G$ is therefore irreducible.

The problem \sdpp is also trivial for irreducible
deterministic presentations,
since every irreducible sofic shift has a synchronizing deterministic presentation by Theorem \ref{thm:irred-char}.
One can compute this synchronizing deterministic
presentation in polynomial time by simply computing the follower-separation $G/{\sim}$.
The presentation $G/{\sim}$ is irreducible as $G$ is, so every vertex in $G/{\sim}$ is reachable
from every other vertex.
As every follower-separated deterministic
presentation has a synchronizing word, by irreducibility, one can extend this word to one that synchronizes to any other vertex.

For deterministic presentations in general, all the problems in Table~\ref{tab:results-overview}, with the
exception of \syncp, remain \cc{PSPACE}-complete when restricted to
follower-separated instances. The reason is those problems ask a
question about the sofic shift a given input presents, and
follower-separation of an input is a polynomial-time operation that
preserves the sofic shift it presents.  For example, given
presentations $G$ and $H$, deciding if $\shift{G} = \shift{H}$ is
equivalent to deciding if $\shift{G/{\sim}} = \shift{H/{\sim}}$, as
$\shift{G} = \shift{G/{\sim}}$ and $\shift{H} = \shift{H/{\sim}}$.

\subsection{Open problems}

Aside from long-standing open problems like the decidability of conjugacy for sofic shifts, our work suggests several interesting open questions pertaining to the size of various objects.
For deterministic presentations in general, the shortest
synchronizing word in a presentation can be exponentially large.
However, for follower-separated determinsitic presentations,
the shortest synchronizing word has at most cubic length with respect
to the number of vertices; for a
follower-separated input to Algorithm \ref{alg:sw}, the
algorithm always finds a synchronizing word, and one
can easily see that the word returned must be at most cubic length.
Actually, by Exercise 3.4.10 of \citet{lind2021introduction}, one
can see that upper bound can be improved to $n(n-1)$, where
$n$ is the number of vertices in the presentation.
To our knowledge, it is open whether this bound is tight.

For a shift space $X$, define the \term{minimum step} of $X$ to be the
minimum $M$ such that $X$ is $M$-step. By \citet{jonoska1996sofic},
every SFT $X$ has a synchronizing deterministic presentation.  Let
$s(X)$ denote the number of vertices the minimal synchronizing
deterministic presentation of $X$. What is the relationship between
$s(X)$ and the minimum step of $X$?  By Remark \ref{re:sft-bound}, we
know that the minimum step of an SFT $X$ is $O(s(X)^2)$.
To our knowledge, it is also open whether this bound it tight.
One lower bound arises from the family of
\term{run-length limited shifts} $\{X_n\}$, which have minimum step $\Omega(s(X_n))$.
We can repeat the same question for the size of a minimal deterministic
presentation. Let $s_d(X)$ denote the number of vertices in a
minimal deterministic presentation of $X$. What is the
relationship between $s_d(X)$ and the minimum step of $X$?  In
Appendix \ref{apx:pspace}, we generalize Remark \ref{re:sft-bound} to
deterministic presentations in general as Proposition
\ref{prop:sft-bound-general}: for a determinstic presentation $G$,
$\shift{G}$ is an SFT if and only if it is $2^{2|Q_G|}$-step, which
implies that the minimum step of an SFT $X$ is $2^{O(s_d(X))}$.
Again, it is open whether this bound is tight.

\bibliographystyle{plainnat}
\bibliography{refs}

\begin{appendices}

\section{Problems in \PSPACE}
\label{apx:pspace}

Here, we show that all the problems in Table~\ref{tab:problems} are in \cc{PSPACE}.
We will rely heavily on Savitch's theorem: if there is a nondeterministic polynomial-space algorithm for a decision problem, then there is a (deterministic) polynomial-space algorithm for it as well \cite{arora2009computational}.

\newcommand{\act}[2]{\llbracket {#1} \rrbracket_{#2}} \newcommand{\then}{\mathbin{;}}

Let $G$ be a deterministic presentation. 
In general, for the decision problems we work with,
given a word $w$, we usually want to know the
value of $Q_G \cdot w$.
In particular, we have the correspondence
$Q_G \cdot w \neq \varnothing$ if and only if $w \in \Bc(\shift{G})$.
In fact, $q \in Q_G \cdot w$ if and only if there is
path labeled $w$ ending at that $q$.
Note the asymmetry of information here: the set of
vertices in $G$ such that there is a path labeled $w$ \emph{starting}
at that vertex is not encoded in $Q_G \cdot w$.
A situation arises because of this asymmetry when designing polynomial-space algorithms
for \sdpp and \sftp when only using the transition action:
one must deduce $Q_G \cdot uw$ given only $Q_G \cdot w$ and $u$, but not $w$, a problem which is generally ill-posed.

To fix this asymmetry, we introduce the \term{action} of a word. The action of
a word $w$ in $G$ is a binary relation $\act{w}{G}$ on the vertices of $G$
such that $(p,q) \in \act{w}{G}$ if and only if there is a path
labeled $w$ from $p$ to $q$. In other words,
\[
  \act{w}{G} \triangleq \{\, (p, q) \in Q^2 : w \in F_G(p) \text{ and } p \cdot w = q \,\}.
\]
Note that $Q_G \cdot w = \{ q \in Q : \exists p \in Q,\, (p,q) \in \act{w}{G}\}$, so the action of a word still encodes the transition action, but also includes more information.
In particular, there is a path labled $w$ ending at $q$ if and only if there is
a vertex $p$ with $(p,q) \in \act{w}{G}$, and
there is a path labeled $w$ starting at $p$ if and only if there is
a vertex $q$ with $(p,q) \in \act{w}{G}$. In fact, we have
$\act{w}{G} \neq \varnothing$ if and only if $w \in \Bc(\shift{G})$.
Observe that $\act{\epsilon}{G} = \{(q, q) : q \in Q_G\}$.

Just as the transition action had a nice algebraic behvaior via
the equation $S \cdot uv = (S \cdot u) \cdot v$, there is an
analagous equation for actions involving the \term{relational composition}
operation. For binary relations $R, S \subseteq Q^2$, define the
relational composition $R \then S$ (pronounced \emph{R then S}) as
\[
  R \then S \triangleq
  \{\, (p, r) \in Q^2 : \exists q \in Q,\, (p, q) \in R \text{ and } (q, r) \in S \,\}.
\]
One can verify that $(p,q) \in \act{u}{G} \then \act{v}{G}$ if and only if
there is a path labeled $uv$ from $p$ to $q$. This implies that
for all $u, v \in \Ac_G^*$, we have $\act{uv}{G} = \act{u}{G} \then \act{v}{G}$.
For example, we can deduce that $\act{\epsilon}{G}$ acts as an identity:
$\act{w}{G} \then \act{\epsilon}{G} = \act{w \epsilon}{G} = \act{w}{G}$.
Algebraically, we can summarize that $\act{\cdot}{G}$ is a semigroup morphism
from $\Ac_G^*$ that recognizes $\Bc(\shift{G})$ \cite{pin2010mathematical}.

We first show that \inclusionp is \cc{PSPACE} by giving a
nondeterministic polynomial-time algorithm for the complement. Given a
deterministic presentations $G$ and $H$, deciding
$\shift{G} \nsubseteq \shift{H}$ is equivalent to
deciding if there exists a word $w$ with $w \in \Bc(\shift{G})$ and $w \notin \Bc(\shift{H})$.
Using $\act{\cdot}{G}$, this is equivalent to deciding if there exists a word $w$ with $\act{w}{G} \neq \varnothing$ and $\act{w}{H} = \varnothing$.
The following nondeterministic algorithm decides the latter predicate: initialize a relation $R$ to $\act{\epsilon}{G}$ and a relation $S$ to $\act{\epsilon}{H}$, and then repeat the following forever: if $R \neq \varnothing$ and $S = \varnothing$, then return true; otherwise, nondeterministically choose some $a \in \Ac_G \cup \Ac_H$ and update $R$ to $R \then \act{a}{G}$ and $S$ to $S \then \act{a}{H}$.

The size of $R$ and $S$ are polynomial with respect to the size
of $G$, so the algorithm is a nondeterminstic polynomial-space algorithm.
\footnote{
  This algorithm does not halt, but in principle, any
  space-bounded algorithm can be modified to halt with a
  logarithmic overhead.
}
Thus, by Savitch's theorem, \inclusionp is in \cc{PSPACE}.
Membership of \equalp in \cc{PSPACE} follows, by testing both $\shift{G} \subseteq \shift{H}$ and $\shift{H} \subseteq \shift{G}$, using the polynomial-space algorithm for \inclusionp twice.

Next, we have that \minimalp is in \cc{PSPACE}: given a deterministic
presentation $G$ and a positive integer $k$,
when $|Q_G| \leq k$, we can always admit a presentation of $\shift{G}$
with $k$ vertices by adding superfluous vertices to $G$. In the case of
$|Q_G| > k$, we can nondeterministically guess a presentation with $k$
vertices (whose size is polynomially bounded as $|Q_G| > k$) and use
the polynomial-space algorithm for \equalp to determine if our guess is a
presentation of $\shift{G}$.

We also have that \irredp is in \cc{PSPACE}: given a
follower-separated deterministic presentation $G$, using the
polynomial-space algorithm for \equalp, we can test if any terminal irreducible
component presents $\shift{G}$.  By Theorem~\ref{thm:irred-char}, such
a terminal irreducible component exists if and only if $\shift{G}$ is
irreducible.

The argument for \sdpp is more complex.
We will break the algorithm into nondeterministic subprocedures, which each perform a particular test.
We can determinize all three with Savitch's theorem, allowing us to use them in further subprocedures.
Let $G$ be a deterministic presentation, and let $R \subseteq Q_G^2$ be
a binary relation.
\begin{itemize}
\item We say $R$ is an action if there is a word $w$ with $\act{w}{G}=R$.
  We denote the set of actions as $\act{\Ac^*_G}{}$.  A simple
  nondeterministic polynomial-space procedure to test if $R$ is an
  action can be implemented by initializing a relation $S$ as
  $\act{\epsilon}{G}$, and in a loop forever: if $S = R$, then return
  true; otherwise, nondeterminically choose $a \in \Ac_G$ and update
  $S$ to $S \then \act{a}{G}$.
\item We say $R$ is intrinsically synchronizing if for every
  $S, T \in \act{\Ac^*_G}{}$, $S \then R \neq \varnothing$ and
  $R \then T \neq \varnothing$ imply
  $S \then R \then T \neq \varnothing$. One can verify that for a word
  $w \in \Bc(\shift{G})$, $w$ is intrinsically synchronizing for
  $\shift{G}$ if and only if $\act{w}{G}$ is intrinsically
  synchronizing. A nondeterministic polynomial-space procedure to test
  if $R$ is not intrinsically synchronizing can be implemented by
  nondeterministically choosing $S, T \subseteq Q^2$, and testing the
  following four predicates: (i) $S$ and $T$ are actions, (ii)
  $S \then R \neq \varnothing$, (iii) $R \then T \neq \varnothing$, and (iv) $S \then R \then T = \varnothing$. If all the tests were true,
  then return true.
\item We say $R$ is preceded by an intrinsically synchronizing word if
  there is some $S \in \act{\Ac^*_G}{}$ that is intrinsically
  synchronizing and $S \then R \neq \varnothing$.
  One can verify that for a word $w$,
  there is a word $u$ that is intrinsically
  synchronizing for $\shift{G}$ with $w \in F_{\shift{G}}(u)$ if
  and only if $\act{w}{G}$ is preceded by an intrinsically synchronizing word.
  A nondeterministic polynomial-space procedure to test if
  $R$ is preceded by an intrinsically synchronizing word can be
  implemented by nondeterministically choose $S \subseteq Q^2$, and testing three predicates: (i) $S$ is an action, (ii) $S$ is intrinsically synchronizing, and (iii) $S \then R \neq \varnothing$. If all the tests were true, then
  return true.
\end{itemize}
With these definitions and by Theorem~\ref{thm:sdp-char}, one
can verify that $\shift{G}$ has
a synchronizing determinsitic presentation if and only if for every
$R \in \act{\Ac^*_G}{}$ with $R \neq \varnothing$, $R$ is preceded
by an intrinsically synchronizing word. Using our subprocedures,
a nondeterministic polynomial-space procedure to test if
$\shift{G}$ does not have a synchronizing deterministic presentation can
be implemented by nondeterministically choosing $R \subseteq Q^2$ and testing whether (i) $R$ is an action, (ii) $R \neq \varnothing$, and (iii) $R$ is not preceded by an intrinsically synchronizing word.

Finally, we show \sftp is in \cc{PSPACE}.
Recall from Remark \ref{re:sft-bound} that
for a follower-separated synchronizing deterministic presentation $G$, $\shift{G}$
is an SFT if and only if it is $(|Q_G|^2-|Q_G|)$-step.
To decide \sftp in polynomial space, we first generalize this characterization
to when $G$ is not necessarily follower-separated and synchronizing.

\begin{proposition}\label{prop:sft-bound-general}
  Let $G$ be a deterministic presentation.
  Then, $\shift{G}$ is an SFT if and only if it is $2^{2|Q_G|}$-step.
\end{proposition}

\begin{proof}
  If $\shift{G}$ is $2^{2|Q_G|}$-step, then it is an SFT.
  Conversely, suppose $\shift{G}$ is an SFT.
  Then, by \citet[Corollary 5.4, Proposition 6.2]{jonoska1996sofic}, $\shift{G}$
  has a follower-separated synchronizing deterministic presentation, and it has at most
  $2^{|Q_G|}$ vertices. Thus, by Remark \ref{re:sft-bound}, $\shift{G}$ is $M$-step
  for some $M \leq (2^{|Q_G|})^2 - (2^{|Q_G|})$. Since
  any shift space that is $M$-step is also $M'$-step for every $M' \geq M$,
  then $\shift{G}$ must be $2^{2|Q_G|}$-step.
\end{proof}

A nondeterministic polynomial-space procedure to test if $\shift{G}$ is not
$2^{2|Q_G|}$-step can be implemented by initializing a counter with
$2 |Q_G|$ bits to $0$ and initializing a relation $R$ to
$\act{\epsilon}{G}$. Then, use the counter to repeat the following
$2^{2|Q_G|}$ times: nondeterministically choose $a \in \Ac_G$ and
update $R$ to $R \then \act{a}{G}$. After the loop,
nondeterministically choose $S \subseteq Q^2$ and test the following three predicates:
(i) $S$ is an action,
(ii) $R \then S \neq \varnothing$, and (iii) $R \then S$ is not
intrinsically synchronizing.
If all the tests were true, then return true.

\section{Intersection Construction}
\label{apx:cap-cons}

To prove Theorem~\ref{thm:min-sync-word-lower-bound}, we adapt the construction from \citet{ang2010problems}.
We construct a family of DFAs $M_{i, k}$ and show that they satisfy the following.
\begin{theorem}
  For every $k \geq 0$, the language $\bigcap_{i=0}^k L(M_{i,k})$ is
  nonempty and the minimum length of a word in
  $\bigcap_{i=0}^k L(M_{i,k})$ is $2^k$.
\end{theorem}

We will define $M_{i,k}$ for $i \geq 0$ and $k \geq 0$.  Each
$M_{i,k}$ has the state set $Q \triangleq \{q_0, q_1, q^*\}$ and are
defined over the alphabet $\{0, 1, \dots, k\}$. The transition
function of $M_{i,k}$ is denoted $\delta_{i,k}$ and is defined as
$\delta_i$ restricted to the domain $Q \times \{0,1,\dots,k\}$,
where $\delta_i:Q \times \mathbb N \to Q$ is defined as follows.
For $j \geq 0$, we define
\begin{gather*}
  \delta_0(q_0, j) \triangleq
  \begin{cases}
    q_1 & \text{if } j = 0 \\
    q^* & \text{otherwise}
  \end{cases}
  \qquad
  \delta_0(q_1, j) \triangleq
  \begin{cases}
    q_1 & \text{if } j \neq 0 \\
    q^* & \text{otherwise}
  \end{cases}
  \qquad
  \delta_0(q^*, j) \triangleq q^*
\end{gather*}
For $i \geq 1$ and $j \geq 0$, we define
\begin{gather*}
  \delta_i(q_0, j) \triangleq
  \begin{cases}
    q_0 & \text{if } j > i \\
    q_1 & \text{if } j < i \\
    q^* & \text{otherwise}
  \end{cases}
  \qquad
  \delta_i(q_1, j) \triangleq
  \begin{cases}
    q_0 & \text{if } j = i \\
    q_1 & \text{if } j > i \\
    q^* & \text{otherwise}
  \end{cases}
  \qquad
  \delta_i(q^*, j) \triangleq q^*
\end{gather*}
We set the initial state of $M_{i,k}$ to $q_0$ for every $i \geq 0$ and $k \geq 0$.
We set the final state of $M_{0,k}$ to just $q_1$,
and for $i \geq 1$, we set the final state of $M_{i,k}$ to just $q_0$.
Figure~\ref{fig:inter-construct-ex} depicts an example of the construction.
Essentially, for $i \geq 1$, we have that $\delta_i$ acts on $Q$
in the following way: $q^*$ is a sink state, so no word that under
that action of $\delta_i$ that visits the sink state will be in the
language of $M_{i,k}$; for $j > i$, we have that $q \mapsto \delta_i(q, j)$
is the identity function; for $j \leq i$, the only transition out of $q_0$
that does not lead to $q^*$ is when $j \neq i$, and
conversely, the only transition out of $q_1$ that does not lead
to $q^*$ is when $j = i$.

Define $h_k\colon \{0,1,\dots,k\}^* \to \{0,1,\dots,k\}^*$ to be the
the string homomorphism $h_k : j\mapsto j k$ that inserts the symbol $k$ after every
symbol from the input.
Define $w_0 \triangleq 0$ and
$w_{k+1} \triangleq h_{k+1}(w_k)$ for $k \geq 0$.  Once can show by
induction for $k \geq 0$ that $w_k \in \bigcap_{i=0}^k L(M_{i,k})$ and
that $|w_k| = 2^k$.  Thus, the minimum length of a word in
$\bigcap_{i=0}^k L(M_{i,k})$ is at most $2^k$.  In fact, as we show
next, any word in $\bigcap_{i=0}^k L(M_{i,k})$ must have length at
least $2^k$, so $w_k$ achieves the minimal length.

\begin{proposition}
  For all $k \geq 0$, if $w \in \bigcap_{i=0}^k L(M_{i,k})$, then $|w| \geq 2^k$.
\end{proposition}

\begin{proof}
  We show this by induction. For $k=0$, any word in $L(M_{0,0})$ must have length
  at least $1 = 2^0$, so the proposition holds for $k=0$.
  
  Now, suppose the proposition holds at a given $k \geq 0$.
  Let $w \in \bigcap_{i=1}^{k+1} L(M_{i, k+1})$. Note that when
  $i \leq k$, then $\delta_{i,k+1}(q, k+1) = q$ for all $q \in Q$.
  Thus, if we let $w'$ denote the word with every occurance of $k+1$ in $w$
  removed, then
  for $i \in \{0, 1, \dots, k\}$, we have $\delta_{i,k+1}(q, w') = \delta_{i,k+1}(q, w)$,
  which implies that $w' \in \bigcap_{i=1}^k L(M_{i,k+1})$. In fact,
  as $w' \in \{0, 1, \dots, k\}^*$ and as $\delta_{i,k+1}(q, j) = \delta_{i,k}(q, j)$
  when $j \in \{0, 1, \dots, k\}$, we have $w' \in \bigcap_{i=0}^k L(M_{i,k})$.
  By our induction hypothesis, we have that $|w'| \geq 2^k$.

  Note that as $w \in L(M_{k+1, k+1})$, then $w$ must alternate between some
  symbol in $\{0, 1, \dots, k\}$ and $k+1$, which implies number of
  $k+1$'s in $w$ must equal the number of $\{0,1,\dots,k\}$'s. Thus,
  as $|w'|$ measures the number of $\{0,1,\dots,k\}$'s in $w$, the number
  of $k+1$'s in $w$ must be at least $2^k$.
  Putting this together, we have $|w| \geq 2\cdot2^k = 2^{k+1}$,
  so the proposition holds at $k+1$.
\end{proof}

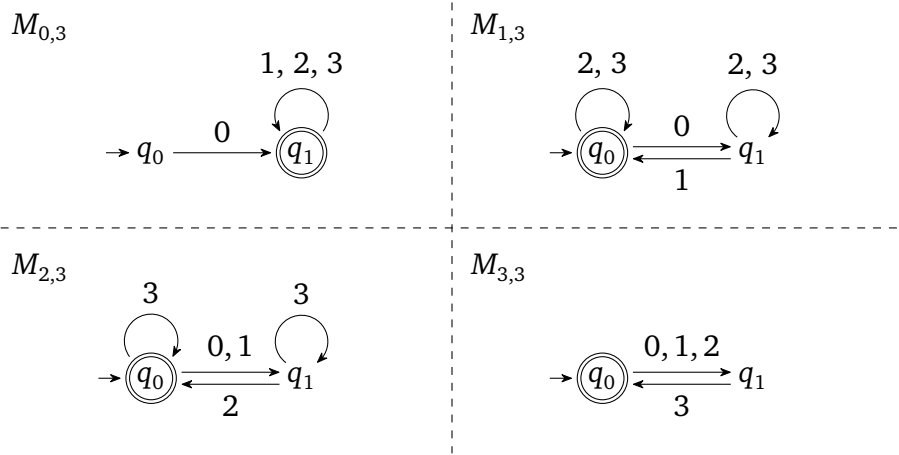
\begin{figure}[t]
  \centering

  \def\agap{.04}
  \begin{tikzpicture}[xscale=\xgap,yscale=\ygap,
    right loop/.style={arc center=right, arc direction=ccw, arc center offset=4pt},
    left loop/.style={arc center=left, arc direction=cw, arc center offset=4pt}]
    
    \path (0,0) node[vertex] (q00) {$q_0$} ++(1, 0)
    node[circle, draw, inner sep=2pt, outer sep=2pt] (q01) {$q_1$}
    node[circle, draw, inner sep=1pt, invisible] {$q_1$};
    
    \path (3,0)
    node[circle, draw, inner sep=2pt, outer sep=2pt] (q10) {$q_0$}
    node[circle, draw, inner sep=1pt, invisible] {$q_0$}
    ++(1, 0) node[vertex] (q11) {$q_1$};
    
    \path (0,-1.5)
    node[circle, draw, inner sep=2pt, outer sep=2pt] (q20) {$q_0$}
    node[circle, draw, inner sep=1pt, invisible] {$q_0$}
    ++(1, 0) node[vertex] (q21) {$q_1$};
    
    \path (3,-1.5)
    node[circle, draw, inner sep=2pt, outer sep=2pt] (q30) {$q_0$}
    node[circle, draw, inner sep=1pt, invisible] {$q_0$}
    ++(1, 0) node[vertex] (q31) {$q_1$};

    \draw (q00) ++(-.3, 0) to (q00);
    \draw (q00) to node[above] {$0$} (q01);
    \draw (q01.north east) to[right loop, arc center offset=3pt]
    node[above] {$1$, $2$, $3$} (q01.north west);

    \draw (q10) ++(-.35, 0) to (q10);
    \draw ($(q10.east)+(0,\agap)$) to node[above] {$0$} ($(q11.west)+(0,\agap)$);
    \draw ($(q11.west)-(0,\agap)$) to node[below] {$1$} ($(q10.east)-(0,\agap)$);
    \draw (q11.north west) to[left loop] node[above] {$2$, $3$} (q11.north east);
    \draw (q10.north west) to[left loop, arc center offset=3pt]
    node[above] {$2$, $3$} (q10.north east);

    \draw (q20) ++(-.35, 0) to (q20);
    \draw ($(q20.east)+(0,\agap)$) to node[above] {$0,1$} ($(q21.west)+(0,\agap)$);
    \draw ($(q21.west)-(0,\agap)$) to node[below] {$2$} ($(q20.east)-(0,\agap)$);
    \draw (q21.north west) to[left loop] node[above] {$3$} (q21.north east);
    \draw (q20.north west) to[left loop, arc center offset=3pt]
    node[above] {$3$} (q20.north east);

    \draw (q30) ++(-.35, 0) to (q30);
    \draw ($(q30.east)+(0,\agap)$) to node[above] {$0,1,2$} ($(q31.west)+(0,\agap)$);
    \draw ($(q31.west)-(0,\agap)$) to node[below] {$3$} ($(q30.east)-(0,\agap)$);

    \draw[-, dashed] (-1, -.5) to (5, -.5);
    \draw[-, dashed] (2, 1) to (2, -2);
    \node[below right] at (-1, 1) {$M_{0,3}$};
    \node[below right] at (2.05, 1) {$M_{1,3}$};
    \node[below right] at (-1, -.6) {$M_{2,3}$};
    \node[below right] at (2.05, -.6) {$M_{3,3}$};
    
  \end{tikzpicture}
  \caption{$M_{i,k}$ for $i \in \{0,1,2,3\}$ and $k=3$.
    The transitions to $q^*$ are not depicted.}
  \label{fig:inter-construct-ex}
\end{figure}

\end{appendices}

\end{document}